\documentclass{LMCS}

\def\dOi{10(3:24)2014}
\lmcsheading%
{\dOi}
{1--28}
{}
{}
{Nov.~25, 2013}
{Sep.~18, 2014}
{}

\ACMCCS{[{\bf Theory of computation}]: Formal languages and automata
  theory---Regular languages} 


\usepackage{amsmath,amssymb,xspace,pdfsync,mathrsfs,graphicx,verbatim}
\usepackage{refcount,hyperref}

\usepackage{tikz}
\usepackage[T1]{fontenc}
\usetikzlibrary{trees,arrows,patterns,shapes,snakes,decorations,fit,shadows,calc}

\usepackage{amsthm}

\usepackage[textwidth=3cm]{todonotes}
\newcommand{\mynote}[3][]{\todo[caption={\sf #3}, color={%
    \ifnum#2=0 green!20
    \else\ifnum#2=1 orange!30
    \else\ifnum#2=2 yellow!20
    \else\ifnum#2=3 cyan!20
    \else magenta!20\fi\fi\fi\fi}, size=\tiny, #1]{\renewcommand{\baselinestretch}{1}\selectfont\sf#3}\xspace}

\newcommand\mypar[1]{\par\medskip\noindent\textbf{#1}}

\definecolor{my1}{cmyk}{0,.6,0,0}
\definecolor{my2}{cmyk}{.3,.0,.0,.0}

\newcommand\nat{\ensuremath{\mathbb{N}}\xspace}
\newcommand\As{\ensuremath{\mathcal{A}}\xspace}

\newcommand\Ps{\ensuremath{\mathcal{P}}\xspace}
\newcommand\Cs{\ensuremath{\mathcal{C}}\xspace}

\newcommand{\lt}{\textup{LT}\xspace}
\newcommand{\ltt}{\textup{LTT}\xspace}
\newcommand{\lttm}{\textup{LTT+MOD}\xspace}

\newcommand\kloop{$k$-loop\xspace}
\newcommand\kloops{$k$-loops\xspace}

\newcommand\image[1]{$#1$-image\xspace}
\newcommand\images[1]{$#1$-images\xspace}
\newcommand\kimage{\image{k}}
\newcommand\kimages{\images{k}}

\newcommand\profile[1]{$#1$-profile\xspace}
\newcommand\profiles[1]{$#1$-profiles\xspace}
\newcommand\kprofile{\profile{k}}
\newcommand\kprofiles{\profiles{k}}

\newcommand\ltteq[2]{\ensuremath{\equiv_{#1}^{#2}}\xspace}
\newcommand\kdltteq{\ltteq{k}{d}}
\newcommand\lteq[1]{\ensuremath{\equiv_{#1}}\xspace}
\newcommand\klteq{\lteq{k}}

\newcommand{\ltclos}[2]{\ensuremath{[#1]_{#2}}}
\newcommand{\lttclos}[3]{\ensuremath{[#1]_{#2}^{#3}}}

\newcommand{\fos}{\ensuremath{\textup{FO}(+1)}\xspace}

\newcommand\decop[1]{\ensuremath{#1}-decomposition\xspace}
\newcommand\decops[1]{\ensuremath{#1}-decompositions\xspace}
\newcommand\pfsdecomp{\decop{\Ps}}
\newcommand\pfsdecomps{\decops{\Ps}}

\newcommand\Sep{\ensuremath{\mathcal{S}}\xspace}
\newcommand\C{\ensuremath{\mathcal{C}}\xspace}

\newcommand\frb{\ensuremath{\mathfrak{b}}\xspace}

\newcommand\frp{\ensuremath{\mathfrak{p}}\xspace}

\newcommand\frs{\ensuremath{\mathfrak{s}}\xspace}

\let\le\leqslant
\let\leq\leqslant
\let\ge\geqslant
\let\geq\geqslant
\theoremstyle{plain}

\newtheorem*{claim}{Claim}

\title[On Separation by LT and LTT]{On Separation by Locally Testable
  and Locally Threshold Testable Languages}

\author{Thomas~Place}
\address{LaBRI, Bordeaux University, France}
\email{firstname.lastname@labri.fr}  %
\thanks{Supported by the Agence Nationale de la Recherche ANR 2010 BLAN 0202 01 FREC}

\author{Lorijn~van Rooijen}

\author{Marc~Zeitoun}

\begin{document}

\begin{abstract}
  A separator for two languages is a third language containing the
  first one and disjoint from the second one. We investigate the
  following decision problem: given two regular input languages,
  decide whether there exists a locally testable (resp.\ a locally
  threshold testable) separator. In both cases, we design a decision
  procedure based on the occurrence of special patterns in automata
  accepting the input languages.  We prove that the problem is
  computationally harder than deciding membership. The correctness
  proof of the algorithm yields a stronger result, namely a
  description of a possible separator.  Finally, we discuss the same
  problem for context-free input~languages.
\end{abstract}

\maketitle

\section{Introduction}
\label{sec:intro}

\makeatletter{}%
\mypar{Context.} The strong connection between finite state devices
and descriptive formalisms, such as first-order or monadic
second-order logic, has been a guideline in computer science since the
seminal work of
B\"uchi~\cite{Buchi:Weak-Second-Order-Arithmetic-Finite:1960:a},
Elgot~\cite{Elgot:Decision-Problems-Finite-Automata:1961:a} and
Trakhtenbrot~\cite{Trakh61}.
This bridge has continuously
been fruitful, disseminating tools and bringing a number of
applications outside of its original research area. For instance,
compiling logical specifications into various forms of automata has
become one of the most successful methods in automatic program
verification~\cite{DBLP:conf/lics/VardiW86}.

One of the challenging issues when dealing with a logical formalism is to
precisely understand its expressiveness and its limitations. While
solutions to \emph{decide} such logics often use a compilation procedure from
formulas to automata, capturing the expressive power amounts to the
opposite translation: given a language, one wants to know whether one
can reconstruct a formula that describes it. In other words, we want
 to solve an instance of the \emph{membership problem}, which asks whether
an input language belongs to some given class.

For regular languages of finite words, the main tool developed to
capture this expressive power is the syntactic
monoid~\cite{pin:hal-00143946}: this is a finite, computable, algebraic
abstraction of the language, whose properties make it possible to
decide  membership. An emblematic example is the membership problem
for the class of first-order definable languages, solved by
Sch\"utzenberger~\cite{Schutzenberger:finite-monoids-having-only:1965:a}
and McNaughton and Papert~\cite{mcnaughton&papert:1971:counter}, which
has led to the development of algebraic methods for obtaining
decidable characterizations of logical
or combinatorial~properties.

\mypar{The separation problem and its motivations.} We consider here the
\emph{separation problem} as a generalization of the membership problem.
Assume we are given two classes of languages $\C$ and \Sep. The
question is, given \emph{two} input languages from $\C$, 
 whether we can separate them by a language from~\Sep. Here, we say
that a language \emph{separates} $K$ from $L$ if it contains $K$ and
is disjoint from $L$.  An obvious necessary condition for separability
is that the input languages $K,L$ be disjoint. A separator language
\emph{witnesses} this condition.

\smallskip
One strong motivation for this problem is to understand the limits of
logics over finite words. Notice that membership reduces to separation
when \C is closed under complement, because checking that a language
belongs to \Sep amounts to testing that it is \Sep-separable from its
complement. Deciding \Sep-separation is also more difficult than
deciding membership in~\Sep, as one cannot rely on algebraic tools
tailored to the membership problem. It may also be computationally
harder, as we shall see in this paper. Thus, solving the separation
problem requires a deeper understanding of \Sep than what is
sufficient to check membership: one not only wants to decide whether
\Sep is powerful enough to \emph{describe} a language, but also to
decide whether it can \emph{discriminate} between two
input~languages. This discriminating power provides more accurate
information than the expressive power.

\mypar{Contributions.}  %
In general, elements of $\C$ cannot always be
separated by an element of~$\Sep$ and %
there is no minimal separator wrt.\  inclusion. We
are interested in the following~questions:
\begin{enumerate}[label=\({\alph*}]
\item can we \emph{decide} whether one can separate two given
  languages of $\C$ by a language of~$\Sep$?
\item what is the \emph{complexity} of this decision problem?
\item if separation is possible, can we \emph{compute} a separator, and at
  which cost?
\end{enumerate}

\smallskip We investigate the separation problem by locally and locally
threshold testable languages. A language is called \emph{locally testable
  (LT)} if membership of a word can be tested by inspecting its prefixes,
suffixes and infixes up to some length (which depends on the language).  The
membership problem for this class was raised by McNaughton and
Papert~\cite{mcnaughton&papert:1971:counter}, and solved independently by
McNaughton and
Zalcstein~\cite{Zalcstein:Locally-testable-languages:1972:a,McNaughton:Algebraic-decision-procedures-local:1974:a}
and by Brzozowski and
Simon~\cite{Brzozowski&Simon:Characterizations-locally-testable-events:1973:a}. If
the input language is given by a deterministic automaton, the membership problem to LT is 
\textsc{Ptime}~\cite{Kim&McNaughton&McCloskey:polynomial-time-algorithm-local:1989:a}.
This class has several generalizations. The most studied one is that of
\emph{locally threshold testable languages (LTT)}, where counting infixes is
allowed up to some threshold. These are the languages definable in \fos,
\emph{i.e.}, first-order logic with the successor relation (but without the
order). Again, membership is
decidable~\cite{Therien&Weiss:Graph-congruences-wreath-products:1985:a}, and
can actually be again tested in
\textsc{Ptime}~\cite{Pin:expressive-power-existential-first:1996:a,Pin:Expressive-power-existential-first-order:2005:a,Trahtman:Algorithm-Verify-Local-Threshold:2001:a}
if the input language is given by a deterministic automaton.

\smallskip
Actually, the decision problem $(a)$ has been rephrased in purely algebraic
terms~\cite{Almeida:SomeAlg:99}: solving the separation problem for a
class \Sep amounts to computing the so-called \emph{2-pointlike sets}
for the algebraic variety corresponding to \Sep. It has been shown
that both the varieties corresponding to \lt and to \ltt\footnote{These
  algebraic varieties are respectively the pseudovariety of semigroups
  \textsf{LSl} of local semilattices, and the semidirect product
  $\mathsf{Acom}*\mathsf{D}$ of commutative and aperiodic semigroups
  with right zero semigroups.} have computable pointlike sets. This is
a consequence
of~\cite{Costa&Nogueira:Complete-reducibility-pseudovariety:2009:a,Costa:Free-profinite-locally-idempotent:2001:a}
for \lt and
of~\cite{Beauquier&Pin:Languages-scanners:1991:a,Straubing:Finite-semigroup-varieties-form:1985:a,Steinberg:98,Steinberg:01}
for \ltt.
However, this approach suffers some drawbacks.
\begin{itemize}
\item First, the proofs are purely algebraic, and they provide no
    insight on the underlying class \Sep of regular languages
  itself. Instead, the proofs are based on algebraic properties of a
  profinite semigroup that depends on~\Sep, and which is in general uncountable.
\item The proofs involve difficult
  results from profinite semigroup theory, and therefore require a
  significant background in algebra and topology: this is the second
  drawback.
 The separation problem is indeed equivalent to showing
  that the topological closure of the input languages wrt.\ the
  profinite topology do intersect. Deciding nonemptiness of such an
  intersection in turn requires a deep understanding of the algebraic
  properties of this profinite semigroup.
\item Finally, this approach only provides a yes/no answer, but no
  description of what an actual separator might be.
\end{itemize}
The present paper alleviates these drawbacks, by only using elementary
pumping arguments, and providing bounds on the parameters defining the
separators.

\smallskip Our results are as follows: we show that separability of regular
languages by \lt and \ltt languages is decidable by \emph{reduction to fixed
  parameters}: for a fixed threshold, we provide a bound on the length of
infixes that define a possible separator. For \ltt-separators, we also provide
a bound for a sufficient threshold. This reduces the problem to a finite
number of candidate separators, and hence entails decidability. We further get
an equivalent formulation on NFAs in terms of forbidden patterns for the
languages to be separable, which yields an \textsc{Nexptime} algorithm. We
also obtain lower complexity bounds: even starting from DFAs, the problem is
\textsc{Np}-hard for \lt and \ltt (while membership is in
\textsc{Ptime}).  Finally, we discuss the separation problem starting from
context-free input languages rather than regular ones.

\smallskip
The main arguments rely on pumping in monoids or automata.  The core
of our proof is generic: we show that if one can find
two words, one in each input language, that are close enough wrt.\ the
class of separators, then the languages  are not separable. 
Here, ``close enough'' is defined in terms of parameters of the input
languages, such as the size of input~NFAs.

\mypar{The separation problem in other contexts.}  While our main concern is
theoretical, let us mention some motivating applications where separation
occurs as a main ingredient.  

In model checking, reachable configurations of a system can be represented by
a language. Separating it from the language representing bad configurations
proves to be effective for verifying safety of a system. Craig interpolation
is a form of separation used in this~context, as well as in type inference,
theorem proving, hardware
specification~\cite{McMillan:Applications-Craig-Interpolants-Model:2005:a,Henzinger&Jhala&Majumdar&McMillan:Abstractions-from-proofs:2004:a}.
In the same line of thought, separation is the core idea of the reachability
algorithm for vector addition systems designed by
Leroux~\cite{LEROUX-TURING100}, who greatly simplified the original
decidability proof~\cite{DBLP:journals/siamcomp/Mayr84} thanks to a difficult separation theorem: he proved that a
recursively enumerable set of separators (namely Presburger definable sets of
configurations) witnesses non-reachability. A related problem addressed also
in automatic verification is to find separators in terms of small or minimal
DFAs~\cite{Gupta08:autom} using learning
algorithms~\cite{Chen09:learning,Leucker:2006}.  Finally, questions in
database theory also motivated separation questions~\cite{martens}.  These
applications justify a systematic study of the separation problem.  It is
therefore surprising that it deserved only little attention and isolated
work~\cite{Choffrut2007274,Choffrut200627,szygram,Hunt:1982:DGP:322307.322317},
even in the restricted, yet still challenging case of regular~languages.

\mypar{Related work for other classes.}  The separation problem has recently
been shown to be \textsc{Ptime}-decidable for the class of piecewise-testable
languages, independently and with different techniques in~\cite{martens}
and~\cite{PvRZ:mfcs}. In the latter paper, it has also been solved using
elementary pumping arguments for the class of unambiguous languages, a widely
studied class that corresponds to languages definable in first-order logic
with only 2 variables. The case of full first-order logic, already shown decidable in
\cite{Henckell:1988}, has been reproved in~\cite{PZ:lics14} with the proof
canvas of the present paper, which extends to separating languages of infinite
words. Recently, the problem has also been shown decidable for separation by
$\Sigma_2$-definable languages (that is, languages definable by first-order
formulas having only 2 alternations and beginning with an existential block).
This information has been in turn used to get decidable characterizations for
higher levels in the first-order quantifier alternation
hierarchy~\cite{PZ:icalp14}.

\mypar{Paper Outline.} We present the necessary background and notation in
Section~\ref{sec:prelims}. The classes \lt and \ltt are defined in
Section~\ref{sec:classdef}. In Section~\ref{sec:boundk}, we present our
results for solving separation when the counting threshold is assumed to be fixed.
This solves separation in particular for the class \lt. In
Section~\ref{sec:ltt}, we show that separation by \ltt languages can be
decided by bounding the counting threshold. We also provide an optimality
result for the bound we obtain. In Section~\ref{sec:comp}, we show upper and lower
bounds for the separation problem by \lt and \ltt languages. In
Section~\ref{sec:cf}, we consider separation by \lt and \ltt languages when the input
languages to be separated are context-free, rather than regular. We finally present
some open problems and further work in Section~\ref{sec:conc}.

\section{Preliminaries}
\label{sec:prelims}
\makeatletter{}%
\noindent {\textbf{Words and Languages.}} We fix a finite alphabet
$A$. We denote by $A^*$ the free monoid over~$A$. The empty word is denoted by~$\varepsilon$. If $w$ is a word, we
set $|w|$ as the \emph{length}, or \emph{size} of $w$. When $w$ is nonempty, we view
$w$ as a sequence of $|w|$ positions labeled over $A$. We number
positions from $0$ (for the leftmost one) to $|w|-1$ (for the
rightmost one).

\smallskip
\noindent {\textbf{Infixes, Prefixes, Suffixes.}} An \emph{infix} of a
word $w$ is a word $w'$ such that $w=u \cdot w' \cdot v$ for some $u,v
\in A^{*}$. Moreover, if $u=\varepsilon$ (resp.~$v=\varepsilon$) we
say that $w'$ is a \emph{prefix} (resp.~\emph{suffix}) of $w$.

Let $0\le x < y \leq |w|$. We write $w[x,y]$ for the infix of $w$ starting at
position $x$ and ending at position $y-1$. By convention, we also set
$w[x,x] = \varepsilon$. Observe that by definition, when $x \le y \le
z$, we have $w[x,z] = w[x,y] \cdot w[y,z]$. 

\smallskip
\noindent {\textbf{Profiles.}} For $k\in\mathbb{N}$, let $k_\ell =
\lfloor{k/2}\rfloor$ and $k_r = k - k_\ell$. A \emph{\kprofile} is a
pair of words $(w_\ell,w_r)$ of lengths at most $k_\ell$ and $k_r$,
respectively. Given $w \in A^{*}$ and $x$ a position of $w$, the
\emph{\kprofile of $x$} is the pair $(w_\ell,w_r)$ defined as follows:
$w_\ell=w[\max(0,x-k_\ell),x]$ and $w_r=w[x,\min(x+k_r, |w|)]$
(see Figure~\ref{fig:profiled}). 
A \kprofile $(w_\ell,w_r)$ \emph{occurs in a word $w$} if there exists
some position $x$ within $w$ whose \kprofile is
$(w_\ell,w_r)$. Similarly, if $n$ is a natural number, we say that
$(w_\ell,w_r)$ \emph{occurs $n$ times in $w$} if there are $n$
distinct positions in $w$ where $(w_\ell,w_r)$ occurs.
We denote by $A_k$ the set of \kprofiles over 
$A$. Note that its size is $|A_k|=|A|^{O(k)}$.
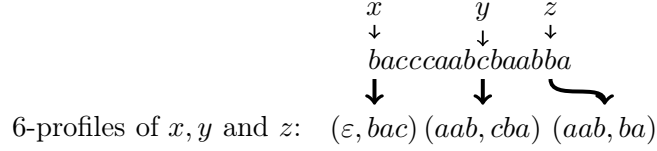
\begin{figure}[h]
  \begin{center}
    \begin{tikzpicture}
      \tikzstyle{ars}=[line width=0.7pt,->]
      \tikzstyle{arp}=[line width=1.5pt,->]
      \tikzstyle{acc}=[line width=1.0pt,snake=brace]
      
      \begin{scope}[scale=.9]
        \node[anchor=mid] (b1) at (0.0,0.0) {$b$};
        \node[anchor=mid] (b2) at (0.2,0.0) {$a$};
        \node[anchor=mid] (b3) at (0.4,0.0) {$c$};
        \node[anchor=mid] (b4) at (0.6,0.0) {$c$};
        \node[anchor=mid] (b5) at (0.8,0.0) {$c$};
        \node[anchor=mid] (b6) at (1.0,0.0) {$a$};
        \node[anchor=mid] (b7) at (1.2,0.0) {$a$};
        \node[anchor=mid] (b8) at (1.4,0.0) {$b$};
        \node[anchor=mid] (b9) at (1.6,0.0) {$c$};
        \node[anchor=mid] (ba) at (1.8,0.0) {$b$};
        \node[anchor=mid] (bb) at (2.0,0.0) {$a$};
        \node[anchor=mid] (bc) at (2.2,0.0) {$a$};
        \node[anchor=mid] (bd) at (2.4,0.0) {$b$};
        \node[anchor=mid] (be) at (2.6,0.0) {$b$};
        \node[anchor=mid] (bf) at (2.8,0.0) {$a$};

        \node[anchor=mid] (x1) at (0.0,0.8) {$x$};
        \node[anchor=mid] (x2) at (1.6,0.8) {$y$};
        \node[anchor=mid] (x3) at (2.6,0.8) {$z$};

        \draw[ars] (x1) to (b1);
        \draw[ars] (x2) to (b9);
        \draw[ars] (x3) to (be);

        \node[anchor=mid] (p1) at (0.0,-1.0) {$(\varepsilon,bac)$};
        \node[anchor=mid] (p2) at (1.6,-1.0) {$(aab,cba)$};
        \node[anchor=mid] (p3) at (3.4,-1.0) {$(aab,ba)$};

        \draw[arp] (b1) to (p1);
        \draw[arp] (b9) to (p2);
        \draw[arp] (be) to [out=-90,in=90] (p3);

        \node[anchor=west] (p0) at ($ (p1) - (5.5,0) $) {\profiles{6} of $x,y$ and $z$:};
      \end{scope}
    \end{tikzpicture}
  \end{center}
  \caption{Illustration of the notion of \kprofile for $k=6$}%
\label{fig:profiled} 
\end{figure}

Intuitively, the \kprofile is the description of the infix of $w$ that
is centered at position~$x$. Observe in particular that the \kprofiles
that occur in a word determine the prefixes and suffixes of length
$(k-1)$ of this word. This is convenient, since we only have to consider
one object instead of three as in the usual presentations of the classes
\lt and \ltt.

\smallskip
\noindent
{\textbf{Separability}.}
Given languages $L,L_1,L_2$ over $A^*$, we say that $L$ \emph{separates} $L_1$
from $L_2$ if 
\[
L_1 \subseteq L \text{ and } L_2 \cap L = \varnothing.
\]
Given a class \Sep of languages, we say that the pair $(L_1,L_2)$ is
\emph{\Sep-separable} if some language $L\in\Sep$ separates $L_1$ from
$L_2$. When \Sep is closed under complement,
$(L_1,L_2)$ is \Sep-separable if and only if $(L_2,L_1)$ is, in which
case we simply say that $L_1$ and $L_2$ are~\Sep-separable.

\smallskip
\noindent
{\textbf{Automata.}}
\label{paragraph:automata}
A \emph{nondeterministic finite automaton} (NFA) over $A$ is denoted by a
tuple $\mathcal{A}=(Q,A,\delta)$, where $Q$ is the finite set of states and
$\delta\subseteq Q\times A\times Q$ is the transition relation.  Observe that this definition
is not the standard one: we do not include the initial and final state
sets (we shall explain below why this is more convenient in our setting).
Abusing notation, we also denote by $\delta\subseteq Q\times A^*\times Q$ the
relation induced by the transition relation.  Given sets $I \subseteq Q$ and
$F\subseteq Q$, we then set
\[
L(\As,I,F) = \{w\in A^* \mid \exists q_I \in I, \exists q_F \in F,\  (q_I,w,q_F)\in \delta\}.
\]
A language $L \subseteq A^*$ is \emph{accepted}, or \emph{recognized}
by $\mathcal{A}$ if \emph{there exist}  sets $I \subseteq Q$ and $F\subseteq
Q$ such that $L=L(\As,I,F)$. In this case,
we call $I$ the set of \emph{initial states} and $F$ the set of \emph{final states for $L$}.
The reason why we do not include initial and final states in our definition of
an NFA is the following: when dealing with separation, we start from two input
languages. However, it is convenient to work with a single automaton
recognizing them both, rather than having to deal with two.  Our definition
permits this: an NFA recognizing both languages can be obtained by building
the Cartesian product of NFAs recognizing the input languages, and by then suitably
choosing initial and final states for each of the input languages.

The \emph{size} $|\mathcal{A}|$ of an automaton $\mathcal{A}$ is its number of states plus its number of
transitions. If $\delta$ is a function, then $\mathcal{A}$ is called a
\emph{deterministic} finite automaton (DFA). 

\smallskip\noindent
{\textbf{Monoids.}}
Let $L$ be a language and let $M$ be a monoid. We say that \emph{$L$ is recognized
  by $M$} if there exists a monoid morphism $\alpha: A^* \rightarrow
M$ together with a subset $F \subseteq M$ such that $L = \alpha^{-1}(F)$. We
also say in this case that $L$ is recognized by $\alpha$. It is
well known that a language is accepted by an NFA if and only if it can be
recognized by a \emph{finite monoid}. We denote by $|M|$ the size of a finite monoid.

In the same way as for NFAs, we want to work with a single monoid recognizing
both input languages. This is easy to obtain, since if $\alpha_1:A^*\to M_1$
(resp.~$\alpha_2:A^*\to M_2$) recognizes $L_1=\alpha_1^{-1}(F_1)$ (resp.~$L_2=\alpha_2^{-1}(F_2)$), then $\alpha:A^*\to M_1\times M_2$ defined by
$\alpha(w)=(\alpha_1(w),\alpha_2(w))$ recognizes both
$L_1=\alpha^{-1}(F_1\times M_2)$ and $L_2=\alpha^{-1}(M_1\times F_2)$.

Finally, it is well known that one can compute from any NFA a finite monoid
recognizing any language that this NFA can accept.  The easiest way to do so is to consider
the \emph{transition monoid} of the NFA, which is generated by Boolean
matrices $M_a\in \{0,1\}^{Q\times Q}$ for $a\in A$, where $M_a(p,q)=1$ if
$(p,a,q)\in\delta$ and $M_a(p,q)=0$ otherwise. It is straightforward to check
that this finite monoid recognizes any language accepted by $\mathcal{A}$.

\section{Locally Testable and Locally Threshold Testable Languages}
\label{sec:classdef}
\makeatletter{}%

In this paper, we investigate two classes of languages. %
Intuitively, a language is locally testable if membership of a word in
the language only depends on the \emph{set} of infixes, prefixes and
suffixes up to some fixed length that occur in the word. For a locally
threshold testable language, membership may also depend on the \emph{number}
of occurrences of such infixes, which may thus be counted up to some
fixed threshold.

\smallskip
In this section we provide specific definitions for both classes. We
start with the larger class of locally threshold testable
languages. In the following, we say that two numbers are \emph{equal
  up to threshold $d$} if either both numbers are equal, or both are
greater than or~equal~to~$d$.

\medskip
\noindent {\textbf{Locally Threshold Testable Languages.}} We say that
a language is \emph{locally threshold testable} (\ltt) if it is a
boolean combination of languages of the form:

\begin{enumerate}
\item\label{p1} $uA^*=\{w \mid \text{$u$ is a prefix of $w$}\}$, for some $u \in A^{*}$.
\item\label{p2} $A^*u=\{w \mid \text{$u$ is a suffix of $w$}\}$, for some $u \in A^{*}$.
\item\label{p3} $\{w \mid \text{$w$ has $u$ as an infix at
    least $d$ times}\}$, for some $u \in A^{*}$ and $d \in \nat$.
\end{enumerate}

\medskip\noindent Actually, \ltt languages can be defined in terms
of first-order logic:  a language is \ltt if and only if it can be
defined by an \fos formula, \emph{i.e.}, a first-order logic formula
using predicates for the equality and next position relations, but not
for the linear
order. See~\cite{Beauquier&Pin:Languages-scanners:1991:a,Thomas:Classifying-regular-events-symbolic:1982:a}.

We also define an index on \ltt languages. Usually, this index is
defined as the smallest size of infixes, prefixes and suffixes needed
to define the language. However, since we only work with \kprofiles,
we directly define an index based on the size of \kprofiles. Given a
\kprofile $(w_\ell,w_r)$, let $|w|_{(w_\ell,w_r)}$ be the number of
positions $x$ in $w$ such that $(w_\ell,w_r)$ is the \kprofile of $x$.
For $w,w'\in A^*$ and $k,d\in\nat$, we write $w \kdltteq w'$
if for every \kprofile $(w_\ell,w_r)$, the numbers $|w|_{(w_\ell,w_r)}$
and $|w'|_{(w_\ell,w_r)}$ are equal up to threshold $d$. 

One can verify that $\kdltteq$ is an equivalence relation (and
actually a congruence) of finite index.  For $k,d \in \nat$, let us denote
by $\ltt[k,d]$ the set of the finitely many languages that are unions
of \kdltteq-classes. By definition,
we have $\ltt=\bigcup_{k,d}\ltt[k,d]$. Given $L\subseteq A^*$, the
smallest $\ltt[k,d]$-language containing~$L$~is
$$\lttclos{L}{k}{d}=\{w\in A^*\mid \exists u\in L\text{ such that }u\kdltteq
w\}.$$
As it is often the case, there is no
smallest \ltt language containing a given regular language.
For instance, over $A=\{a\}$, any \ltt language containing $(aa)^*$
is of the form $a^{2p}a^*\cup F$ with $F$~finite.
Removing $a^{2p+1}$ from such a language yields a smaller \ltt
one, still containing~$(aa)^*$.

\medskip
\noindent
{\textbf{Locally Testable Languages.}} The class of locally
testable languages is the restriction of \ltt languages in which
infixes cannot be counted. A language 
is \emph{locally testable} (\lt) if it is a boolean combination of
languages of the form~\eqref{p1}, \eqref{p2} and the following
restriction of~\eqref{p3}:

\begin{enumerate}
\setcounter{enumi}{3}
\item\label{p4} $ A^* u A^* = \{w \mid \text{$w$ has $u$ as an infix}\}$, for some $u \in A^{*}$.
\end{enumerate}

No simple description of \lt in terms of first-order logic is known.
In terms of linear temporal logic, \lt languages are exactly those defined by
formulas involving only the operators \textsf{F} (eventually) and
\textsf{X} (next), with no nesting of \textsf{F} operators.

Given two words $w,w'$ and a number $k$, we write $w \klteq w'$ 
for $w \ltteq{k}{1} w'$. For all $k \in \mathbb{N}$, we denote by $\lt[k]$ the 
set of languages that are unions of \klteq-classes, and
$\lt=\bigcup_{k}\lt[k]$. Given $L\subseteq A^*$ and $k\in\nat$, the
smallest $\lt[k]$-language containing~$L$ is
$$\ltclos{L}{k}=\{w\in A^*\mid \exists u\in L\text{ such that }u\klteq
w\}.$$
However, as for \ltt, there is no smallest \lt language containing a given
regular language.

\section{Separation for a Fixed Threshold}
\label{sec:boundk}
\makeatletter{}%
In this section, we prove that if $d$ is a fixed natural number,
it is decidable %
whether two languages can be separated by an
\ltt language of counting threshold $d$ (\emph{i.e.}, by an
$\ltt[k,d]$ language for some $k$). In particular, this covers the
case of \lt, which corresponds to $d=1$. All results in this section are
for an arbitrary fixed~$d$. Our result is twofold.

\begin{itemize}
\item First, we establish a bound $k$ on the size of profiles that it 
  suffices to consider in order to separate the input languages. This
  bound only depends on the size of monoids recognizing these languages,
  and it can be computed. One can then use a brute-force algorithm
  that tests separability by all the finitely many $\ltt[k,d]$
  languages.
\item The second contribution is a criterion on the input languages to
  check separability by an $\ltt[k,d]$ language for some $k$. This
  criterion can be defined equivalently on automata or monoids
  recognizing the input languages, in terms of the absence of common patterns.
\end{itemize}

\noindent The section is divided in four subsections. Our criterion is
stated in the first one. The remaining subsections are then devoted
to the statement and proof of the theorem. %

\subsection{Patterns}

In this section, we define our criterion that two languages must satisfy
in order to be separable. The criterion can be defined equivalently on
automata or monoids recognizing the languages.

\medskip\noindent
{\bf Block Patterns.} A \emph{block} is a triple of words
$\frb = (v_\ell,u,v_r)$ where $v_\ell,v_r$ are nonempty. Similarly, a
\emph{prefix block} is a pair of words $\frp=(u,v_r)$ with $v_r$
nonempty, and a \emph{suffix block} is a pair of words
$\frs=(v_\ell,u)$ with $v_\ell$ nonempty. Let $d\in\nat$. A
\emph{$d$-pattern} \Ps is
\begin{itemize}
\item either a word $w$,
\item or a triple $(\frp,f,\frs)$ where $\frp$ and
  $\frs$ are respectively a prefix and a suffix block, and $f$ is a
  function mapping blocks to the set $\{0,\dots,d\}$, such that all
  but finitely many blocks are mapped to $0$.
\end{itemize}
\medskip
\noindent
{\bf Decompositions.} Let $w$ be a word and let $\Ps$ be a $d$-pattern. We
say that \emph{$w$ admits a \pfsdecomp} if  $w$ admits a
decomposition $w = u_0v_1u_1v_2 \cdots v_nu_n$ with $n \geq 0$ and
such that either $n=0$ and $\Ps=u_0=w$, or $\Ps=(\frp,f,\frs)$ and the
following conditions are verified:
\begin{enumerate}
\item $\frp =(u_0,v_1)$ and $\frs = (v_n,u_n)$.
\item for every block $\frb$, if $f(\frb)<d$, then 
  $|\{i\mid(v_i,u_i,v_{i+1})=\frb\}|=f(\frb)$.
\item for every block $\frb$, if $f(\frb)=d$, then
  $|\{i\mid(v_i,u_i,v_{i+1})=\frb\}|\geq d$.
\end{enumerate}

We may say \pfsdecomp to mean \pfsdecomp \emph{of some word}.  Let $\alpha :
A^* \rightarrow M$ be a morphism into a monoid~$M$, and let $s \in M$.  A
\pfsdecomp of $w$ is said to be \emph{$(\alpha,s)$-compatible} if
$\alpha(w)=s$ and $\alpha(u_0\cdots v_i)=\alpha(u_0\cdots v_i) \cdot
\alpha(v_i)$, for all $1 \le i \leq n$. Similarly, if $p,q$ are two states of
an automaton \As, a \pfsdecomp is \emph{$(p,q)$-compatible} if there is a run
from $p$ to $q$ on $w$, such that for all $1 \le i \leq n$, each infix $v_i$
labels a loop in the~run, as pictured in~\figurename~\ref{fig:compat}.
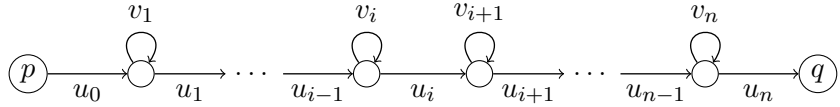
\begin{figure}[h]
  \tikzstyle{nod}=[minimum size=0.35cm,draw,circle,inner sep=2pt]
  \tikzstyle{nof}=[minimum size=0.35cm,draw,circle,double,double distance=1pt]
  \tikzstyle{ars}=[line width=0.5pt,->]
  \centering
  \begin{tikzpicture}
    \node[nod] (A1) at (0.0,0.0) {$p$};
    \node[nod] (A2) at (1.5,0.0) {};
    \node (A3) at (3.0,0.0) {\ldots};
    \node[nod] (A4) at (4.5,0.0) {};
    \node[nod] (A5) at (6.0,0.0) {};
    \node (A6) at (7.5,0.0) {\ldots};
    \node[nod] (A7) at (9.0,0.0) {};
    \node[nod] (A8) at (10.5,0.0) {$q$};

    \draw[ars] (A1) to node[below] {$u_0$} (A2);
    \draw[ars] (A2) to node[below] {$u_1$} ($(A2)!3/4!(A3)$);
    \draw[ars] (A3) to node[below] {$u_{i-1}$} (A4);
    \draw[ars] (A5) to node[below] {$u_{i+1}$} ($(A6)!1/4!(A5)$);
    \draw[ars] ($(A7)!3/4!(A6)$) to node[below] {$u_{n-1}$} (A7) ;

    \draw[ars] (A4) to node[below] {$u_i$} (A5);

    \draw[ars] (A7) to node[below] {$u_n$} (A8);

    \draw[ars] (A2) to [out=125,in=55,loop] node[above] {$v_1$} (A2);
    \draw[ars] (A4) to [out=125,in=55,loop] node[above] {$v_i$} (A4);
    \draw[ars] (A5) to [out=125,in=55,loop] node[above] {$v_{i+1}$} (A5);
    \draw[ars] (A7) to [out=125,in=55,loop] node[above] {$v_n$} (A7);

  \end{tikzpicture}
  \caption{A $(p,q)$-compatible \pfsdecomp $u_0v_1 \cdots v_nu_n$
    \smaller (edges denote transition sequences)}
  \label{fig:compat}
\end{figure}

\medskip\noindent {\bf Common Patterns.} Let $d \in \nat$ and $\alpha:A^*\to
M$ be a morphism into a finite monoid. We say that a pair $(s_1,s_2)\in
M\times M$ has a \emph{common $d$-pattern} if there exist a $d$-pattern \Ps
and two \pfsdecomps of (possibly different) words that are respectively
$(\alpha,s_1)$-compatible and $(\alpha,s_2)$-compatible. In this
terminology, $\alpha$ is understood and not mentioned
explicitly. Similarly, if $\As$ is an automaton, and $p_1,q_1,p_2,q_2$ are states of $\As$, we say that the
pair $\bigl((p_1,q_1),(p_2,q_2)\bigr)$ has a \emph{common $d$-pattern} if
there exist a $d$-pattern \Ps and two \pfsdecomps of words that are
respectively $(p_1,q_1)$-compatible and $(p_2,q_2)$-compatible.  In
particular, $\bigl((p_1,q_1),(p_2,q_2)\bigr)$ has a common 1-pattern if there
are paths in $\As$ of the form shown in~\figurename~\ref{fig:compat}
with the same \emph{set} of triples $(v_i,u_i,v_{i+1})$, going respectively
from $p_1$ to $q_1$ and from $p_2$ to $q_2$.

\smallskip

\subsection{Separation Theorem for a Fixed Threshold} The reason for
introducing common patterns is the following. First, having a common
$d$-pattern for a pair $(s_1,s_2)\in M\times M$ is a decidable
property. Second, it is a necessary and sufficient condition for the languages
$\alpha^{-1}(s_1)$ and $\alpha^{-1}(s_2)$ \emph{not} being separable by any
$\ltt[k,d]$ language, for any~$k$. A similar statement holds for common
$d$-patterns in NFAs. This is what we state now in our main theorem for this
section.

\begin{thm}
  \label{thm:seplttd}
  Fix $d \in \nat$. Let $L_1,L_2$ be regular languages. Let
  $\alpha:A^*\to M$ be a morphism into a finite monoid $M$ recognizing both
  $L_1$ and $L_2$. Let $\As$ be an NFA recognizing both $L_1$ and $L_2$, with
  $L_i=L(\As,I_i,F_i)$. Set $k=4(|M|+1)$. Then, the following conditions
  are equivalent:
  \begin{enumerate}
  \item\label{item:a} $L_1$ and $L_2$ are $\ltt[\ell,d]$-separable for some $\ell$.
  \item\label{item:b} $L_1$ and $L_2$ are $\ltt[k,d]$-separable.
  \item\label{item:c} The language $\lttclos{L_1}{k}{d}$ separates $L_1$ from $L_2$.
  \item\label{item:d} No pair in $\alpha(L_1)\times\alpha(L_2)$ has a common $d$-pattern.
  \item\label{item:e} No pair in $(I_1\times F_1) \times(I_2\times F_2)$  has a common $d$-pattern.
  \end{enumerate}
\end{thm}

Observe that Item~\eqref{item:b} is essentially a \emph{delay theorem}~\cite{Straubing:Finite-semigroup-varieties-form:1985:a}
for separation restricted to the case of \ltt: we prove that the size
of profiles (\emph{i.e.}, infixes) that a potential separator needs to
consider can be bounded by a function of the size of the monoids
recognizing the~languages.
By restricting Theorem~\ref{thm:seplttd} to the case
$d=1$, we get the following separation theorem for \lt, that we explicitly
state in view of the relevance of this class.

\begin{thm}
  \label{thm:seplt}
  Let $L_1,L_2$ be regular languages. Let $\alpha:A^*\to M$ be a
  morphism into a finite monoid $M$ recognizing both $L_1$ and $L_2$. Let $\As$ be an NFA
  recognizing both $L_1$ and $L_2$, with $L_i=L(\As,I_i,F_i)$. Set
  $k=4(|M|+1)$. The following conditions are equivalent:
  \begin{enumerate}
  \item\label{item:a2} $L_1$ and $L_2$ are $\lt$-separable.
  \item\label{item:b2} $L_1$ and $L_2$ are $\lt[k]$-separable.
  \item\label{item:c2} The language $\ltclos{L_1}{k}$ separates
    $L_1$ from $L_2$.
  \item\label{item:d2} No pair in $\alpha(L_1)\times\alpha(L_2)$ has a common $1$-pattern.
  \item\label{item:e2} No pair in $(I_1\times F_1) \times(I_2\times F_2)$  has a common $1$-pattern.
  \end{enumerate}
\end{thm}

Theorem~\ref{thm:seplttd} and Theorem~\ref{thm:seplt} yield algorithms
for deciding \ltt- and \lt-separability for a fixed threshold. Indeed,
the algorithm just tests all the finitely many $\ltt[k,d]$ languages
as potential separators. This brute-force approach yields a very
costly procedure. It turns out that a more practical algorithm can be
obtained from Items~\eqref{item:d} and~\eqref{item:e}. We postpone
the presentation of this algorithm to Section~\ref{sec:comp}.

\begin{cor} \label{cor:decidltd}
Let $d \in \nat$. It is decidable whether two given regular languages
are $\ltt[\ell,d]$-separable for some $\ell \in \nat$. In particular,
it is decidable whether they are $\lt$-separable.
\end{cor}

It remains to prove Theorem~\ref{thm:seplttd}. The implications
$\eqref{item:c} \Rightarrow \eqref{item:b} \Rightarrow \eqref{item:a}$
are immediate by definition.  We now prove the implications
$\eqref{item:a} \Rightarrow \eqref{item:e}\Rightarrow\eqref{item:d}
\Rightarrow \eqref{item:c}$, devoting a separate subsection to the proof of each implication.

\subsection{\texorpdfstring{Implication $\eqref{item:a} \Rightarrow \eqref{item:e}$ in
  Theorem~\ref{thm:seplttd}}{{Implication  (\ref{item:a}) => (\ref{item:e}) in
  Theorem \ref{thm:seplttd}}}}

We prove the contrapositive of $\eqref{item:a} \Rightarrow
\eqref{item:e}$: if $\bigl((p_1,q_1),(p_2,q_2)\bigr)\in(I_1\times F_1) \times(I_2\times F_2)$ has a common $d$-pattern, then
there exists no $\ell \in \nat$ such that $L_1$ and $L_2$ are 
$\ltt[\ell,d]$-separable. This is an immediate consequence of the next
proposition.

\begin{prop} \label{prop:decomptokp} Fix $d \in \nat$ and let
  $\As$ be an NFA. Let $p_1,q_1,p_2,q_2$ be states of $\As$. If $\bigl((p_1,q_1),(p_2,q_2)\bigr)$ has a common
  $d$-pattern, then, for all $\ell \in \nat$, there exist $w_1,w_2$ accepted
  respectively by $L(\As,\{p_1\},\{q_1\})$ and $L(\As,\{p_2\},\{q_2\})$ such that $w_1
  \ltteq{\ell}{d} w_2$.
\end{prop}

\begin{proof}
  Set $L_1=L(\As,\{p_1\},\{q_1\})$ and $L_2=L(\As,\{p_2\},\{q_2\})$.
  Let \Ps be a common $d$-pattern of $\bigl((p_1,q_1),(p_2,q_2)\bigr)$. If $\Ps = w
  \in A^*$, then by definition, $w\in L_1\cap L_2$, and it suffices to choose
  $w_1 = w_2 = w$. Otherwise, $\Ps=(\frp,f,\frs)$ and there exist $z_1\in
  L_1$, $z_2\in L_2$ having a $(p_1,q_1)$- respectively $(p_2,q_2)$-compatible
  \pfsdecomp. Let $z_1=u_0v_1u_1v_2 \cdots v_{n}u_n$ and $z_2=x_0y_1x_1y_2
  \cdots y_{m}x_m$ be these decompositions. For $\ell \in \nat$,~set
  \[
  \begin{array}{lll}
    w_1 & = & u^{}_0v_1^{\ell(d+1)}u^{}_1v_2^{\ell(d+1)} \cdots
    v_{n}^{\ell(d+1)}u^{}_n, \\
    w_2 & = & x^{}_0y_1^{\ell(d+1)}x^{}_1y_2^{\ell(d+1)} \cdots
    y_{m}^{\ell(d+1)}x^{}_m.
  \end{array}
  \]
  By definition of compatibility, we deduce that $w_1\in L_1$ and $w_2\in
  L_2$. We claim that $w_1 \ltteq{\ell}{d} w_2$.  Indeed, from the
  \pfsdecomps of $z_1$ and $z_2$, we deduce that $u_0=x_0$ and
  $v_1=y_1$, which implies that $w_1$ and $w_2$ have the same prefix
  of length $\ell-1$ (and actually, even of length $\ell(d+1)$). Similarly,
  they have the same suffix of length $\ell-1$. To show the claim, it
  remains to verify that each word of length at most $\ell$ occurs the
  same number of times, up to threshold~$d$, as an infix in $w_1$ and
  in $w_2$. Let $u$ be an infix of length at most $\ell$ of, say, $w_1$. 

  Assume first that $u$ occurs in some $v_i^{\ell(d+1)}$. Then it occurs
  at least $d$ times. Since the decompositions of $w_1,w_2$ are
  \pfsdecomps, there exists some $j$ such that $y_j=v_i$. Therefore,
  $u$ occurs at least $d$ times as an infix in~$y_j^{\ell(d+1)}$, hence
  also in $w_2$.

  Assume finally that $u$ does not occur in any
  $v_i^{\ell(d+1)}$. Therefore, it overlaps with some of the $u_i$'s, and
  for these indices $i$, it is an infix of
  $v_i^{\ell(d+1)}u^{}_iv_{i+1}^{\ell(d+1)}$. Since both decompositions of
  $w_1$ and $w_2$ are \Ps-compatible, the number of triples
  $(v_i,u_i,v_{i+1})$ in the decomposition of $w_1$ and the
  number of triples $(y_j,x_j,y_{j+1})$ in that of $w_2$ which are equal to a given
  triple is the same, up to threshold~$d$. Therefore, $u$ occurs the
  same number of times up to threshold $d$ in both $w_1$ and
  $w_2$. We have thus shown that $w_1 \ltteq{\ell}{d} w_2$.
\end{proof}

\subsection{\texorpdfstring{Implication $\eqref{item:e} \Rightarrow \eqref{item:d}$ in
  Theorem~\ref{thm:seplttd}}{{Implication  (\ref{item:e}) => (\ref{item:d}) in
  Theorem \ref{thm:seplttd}}}}

We prove the contrapositive: if there is a pair
$(s_1,s_2)\in\alpha(L_1)\times\alpha(L_2)$ having a common $d$-pattern,
then there is a pair $\bigl((p_1,q_1),(p_2,q_2)\bigr)\in(I_1\times F_1)
\times(I_2\times F_2)$ also having a common $d$-pattern.  This follows from
the following claim, which states that the presence of a ``recognizing pair''
$(s_1,s_2)$ that has a common $d$-pattern does not depend on the choice of the
recognizing monoid morphisms.
\begin{claim}
  Let $d\in\nat$, and let $\alpha:A^*\to M$ and
  $\beta:A^*\to N$ be monoid morphisms recognizing both $L_1$ and $L_2$. If there exists
  $(s_1,s_2)\in \alpha(L_1)\times \alpha(L_2)$ having a common
  $d$-pattern, then there exists $(t_1,t_2)\in \beta(L_1)\times
  \beta(L_2)$ also having a common $d$-pattern.
\end{claim}

Let us admit the claim for a moment. Assume that
$(s_1,s_2)\in\alpha(L_1)\times\alpha(L_2)$ has a common $d$-pattern. Let $N$ be the transition monoid of $\As$, and let
$\beta:A^*\to N$ be the associated morphism. Since $\beta$ recognizes
$L_i$ (see~\cite{pin:hal-00143946}), one can apply the claim: it follows that there exists
$(t_1,t_2)\in\beta(L_1)\times\beta(L_2)$ having a common $d$-pattern. By
definition of a transition monoid, it is then immediate to build from
$(t_1,t_2)$ a pair $\bigl((p_1,q_1),(p_2,q_2)\bigr)\in(I_1\times F_1)
\times(I_2\times F_2)$ having a common $d$-pattern.

\medskip It remains to prove the claim.  Let $\alpha:A^*\to M$
and $\beta:A^*\to N$ be morphisms recognizing both $L_i$ for $i=1,2$, and let $F_i=\alpha(L_i)$
and $G_i=\beta(L_i)$.  Let $\Ps$ be a common $d$-pattern of $(s_1,s_2)\in
F_1\times F_2$. If $\Ps = w \in A^*$, then by definition $w \in
\alpha^{-1}(F_1) \cap \alpha^{-1}(F_2) = L_1\cap L_2=\alpha^{-1}(G_1) \cap
\alpha^{-1}(G_2)$, so \Ps is a common $d$-pattern of $(\beta(w),\beta(w))\in
G_1\times G_2$. Otherwise, $\Ps$ is of the form $(\frp,f,\frs)$. We define a
new $d$-pattern $\Ps'$ that is common to some $(t_1,t_2)\in G_1\times
G_2$. Let $\omega \in \nat$ be such that $s^\omega$ is idempotent for all $s
\in M$ and $s \in N$. For any block $\frb = (v_{\ell},u,v_r)$, set
$\tilde{\frb} = ((v_{\ell})^\omega,u(v_r)^\omega, (v_r)^\omega)$. The mapping
$\frb \mapsto \tilde{\frb}$ is clearly injective. Set $\Ps'$ as the
$d$-pattern $(\frp',f',\frs')$ defined as follows:
\begin{itemize}
\item $\frp' = (u(v_r)^\omega,(v_r)^\omega)$ with $(u,v_r) = \frp$.
\item For any block $\frb'$, if $\frb' = \tilde{\frb}$ for some block
  $\frb$, then $f'(\frb') = f(\frb)$. Otherwise $f'(\frb') = 0$.
\item $\frs' = ((v_\ell)^{\omega},u)$ with $(v_\ell,u) = \frs$.
\end{itemize}
We now prove that $\Ps'$ is a common $d$-pattern of some $(t_1,t_2)\in
G_1\times G_2$. By definition of \Ps, there exist $z_1 \in
\alpha^{-1}(s_1)\subseteq L_1$, $z_2\in\alpha^{-1}(s_2)\subseteq L_2$ having an $(\alpha,
s_1)$-, respectively $(\alpha,s_2)$-compatible
\pfsdecomp. Let $z_1=u_0v_1u_1v_2 \cdots v_{n}u_n$ and
$z_2=x_0y_1x_1y_2 \cdots y_{m}x_m$ be these decompositions. Set
$\tilde u_i = u_iv_{i+1}^\omega$ for $i<n$, $\tilde x_i=x_iy_{i+1}^\omega$ for $i < m$, $\tilde u_n = u_n$, $\tilde x_m=x_m$, $\tilde
v_i=v_i^\omega$ for $i \leq n$, and $\tilde y_i=y_i^\omega$ for $i\leq m$.
Let then
\begin{equation}
\label{eq:newdec}
\begin{array}{lll}
w_1 & = & \tilde u_0\tilde v_1\tilde u_1\tilde v_2 \cdots \tilde v_{n}\tilde u_n, \\
w_2 & = & \tilde x_0\tilde y_1\tilde x_1\tilde y_2 \cdots \tilde y_{m}\tilde x_m.
\end{array}
\end{equation}
It is immediate by definition of $\Ps'$ and $z_1,z_2$ that $w_1,w_2$ are
\decops{\Ps'}. Then, it follows from the definition of $\omega$ for elements of
$M$ that $\alpha(w_i) \in F_i$ for $i \in \{1,2\}$. Hence, we have
$\beta(w_i) \in G_i$. We choose $t_i=\beta(w_i)$, so that $(t_1,t_2)\in
G_1\times G_2$. Finally, it follows from
the definition of $\omega$ for $N$ that for $i \in \{1,2\}$, the
decomposition of $w_i$ given by~\eqref{eq:newdec} is
$(\beta,\beta(w_i))$-compatible. We have thus shown that
$(t_1,t_2)\in G_1 \times G_2$ has a common $d$-pattern.\qed

\subsection{\texorpdfstring{Implication $\eqref{item:d} \Rightarrow
    \eqref{item:c}$ in Theorem~\ref{thm:seplttd}}{Implication (\ref{item:d})
    => (\ref{item:c}) in Theorem~\ref{thm:seplttd}}} Again, we prove the
contrapositive of the statement: if $\lttclos{L_1}{k}{d}$ does not separate
$L_1$ from $L_2$ when $k=4(|M|+1)$, then there exists $(s_1,s_2)\in
\alpha(L_1)\times\alpha(L_2)$ having a common $d$-pattern. Observe that by
definition of $\lttclos{L_1}{k}{d}$, if $\lttclos{L_1}{k}{d}$ does not
separate $L_1$ from $L_2$, then there exist $w_1 \in L_1$, $w_2\in L_2$ with
$w_1 \kdltteq w_2$. We take $s_1=\alpha(w_1)$ and $s_2=\alpha(w_2)$. The
next proposition shows that the pair $(s_1,s_2)$ indeed has a common $d$-pattern.

\begin{prop} \label{prop:ktodecomp} Let $\alpha:A^*\to M$
 be a morphism and let $k=4(|M|+1)$. Let $d
  \in \nat$ and let $w_1,w_2$ be words such that $w_1 \kdltteq
  w_2$. Then, there exists a $d$-pattern \Ps, an
  $(\alpha,\alpha(w_1))$-compatible \pfsdecomp, and an
  $(\alpha,\alpha(w_2))$-compatible \pfsdecomp.
\end{prop}

The remainder of the section is now devoted to the proof of
Proposition~\ref{prop:ktodecomp}. We set $w_1,w_2,k$ and $d$ as in the
statement of the proposition. Observe first that if $w_1=w_2=w$, then it suffices to
take $\Ps=w$. Therefore, we suppose for the remainder of the proof
that $w_1 \neq w_2$. We proceed as follows: we construct two new words
$w'_1,w'_2$ from $w_1,w_2$ such that $w'_1$ admits an
$(\alpha,\alpha(w_1))$-compatible \pfsdecomp and $w'_2$ admits an
$(\alpha,\alpha(w_2))$-compatible \pfsdecomp, for some $d$-pattern
$\Ps=(\frp,f,\frs)$.

\smallskip
We first describe the construction of $w'_1,w'_2$, and then prove that
it is correct. It amounts to duplicating infixes verifying special
properties in $w_1,w_2$. We first define these special infixes,
called \emph{\kloops.}

\medskip\noindent {\bf \kloops.} Let $w \in A^{*}$, $x$ be a
position in $w$, and $(w_\ell,w_r)$ be the \profile{\lfloor
  k/2\rfloor} of $x$. We say that \emph{$x$ admits a \kloop} if there
exists a nonempty prefix $u$ of $w_r$ such that $\alpha(w_\ell) =
\alpha(w_\ell \cdot u)$.  In this case, we call the smallest such $u$ \emph{the
  \kloop of $x$}. See~Figure~\ref{fig:kloop}.

\tikzstyle{ars}=[line width=0.7pt,->]
\tikzstyle{arp}=[line width=1.5pt,->]
\tikzstyle{acc}=[line width=1.0pt,snake=brace]
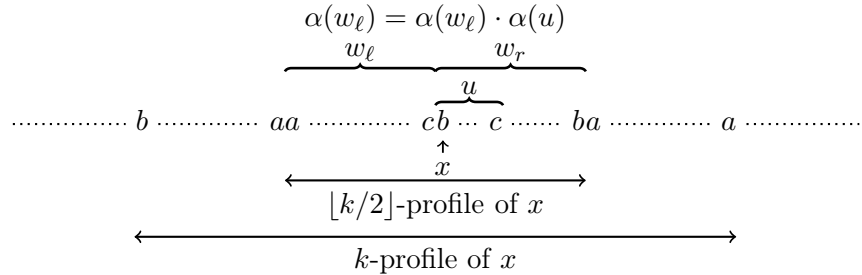
\begin{figure}[h]
  \begin{center}
    \begin{tikzpicture}
      \node[anchor=mid] (b1) at (-3.9,0.0) {$b$};
      \node[anchor=mid] (b2) at (-2.1,0.0) {$a$};
      \node[anchor=mid] (b3) at (-1.9,0.0) {$a$};
      \node[anchor=mid] (b4) at (-0.1,0.0) {$c$};
      \draw[dotted,thick] (b1.mid east) to (b2.mid west);
      \draw[dotted,thick] (b3.mid east) to (b4.mid west);
      \draw[dotted,thick] ($(b1.mid west)-(1.5,0.0)$) to (b1.mid west);
      \node[anchor=mid] (a1) at (0.1,0.0) {$b$};
      \node[anchor=mid] (c1) at (0.8,0.0) {$c$};
      \node[anchor=mid] (a2) at (1.9,0.0) {$b$};
      \node[anchor=mid] (a3) at (2.1,0.0) {$a$};
      \node[anchor=mid] (a4) at (3.9,0.0) {$a$};

      \draw[dotted,thick] (a1.mid east) to (c1.mid west);
      \draw[dotted,thick] (c1.mid east) to (a2.mid west);
      \draw[dotted,thick] (a3.mid east) to (a4.mid west);

      \draw[dotted,thick] (a4.mid east) to ($(a4.mid east)+(1.5,0.0)$);

      \node[anchor=mid] (z1) at (0.1,-0.6) {$x$};

      \draw[ars] (z1) to (a1);

      \draw[<->,thick] (-2.0,-0.75) to node[below] {\profile{\lfloor k/2\rfloor} of $x$} (2.0,-0.75);
      \draw[<->,thick] (-4.0,-1.5) to node[below] {\profile{k} of $x$} (4.0,-1.5);

      \draw[acc] (-2.0,0.7) to node[above] {$w_\ell$} (0.0,0.7);
      \draw[acc] (0.0,0.7) to node[above] {$w_r$} (2.0,0.7);

      \draw[acc] (0.0,0.25) to node[above] {$u$} (0.9,0.25);

      \node at (0.0,1.4) {$\alpha(w_\ell)=\alpha(w_\ell)
        \cdot \alpha(u)$};
    \end{tikzpicture}
  \end{center}
  \caption{A position $x$ admitting a \kloop $u$, that is: $\alpha(w_\ell) =
    \alpha(w_\ell \cdot u)$}
  \label{fig:kloop}
\end{figure}

\noindent For our construction to work, we need \kloops to have three specific
properties. The first two are simple facts that are
immediate from the definition: \kloops are determined by profiles and
can be duplicated without modifying the image of the word under~$\alpha$.

\begin{fact} \label{fct:profloop}
  Let $x$ be a position. Whether $x$ admits a \kloop, and if so, which \kloop $x$ admits, 
  only depends on the \profile{\lfloor k/2\rfloor} of $x$.
\end{fact}

\begin{fact} \label{fct:tloop}
  Let $w$ be a word and let $x$ be a position within $w$ such that $x$ admits a
  \kloop $u$. Then we have $\alpha(w[0,x]) = \alpha(w[0,x])
  \cdot \alpha(u)$. 
\end{fact}

The last property we need is that \kloops occur frequently in words,
\emph{i.e.}, at least one of $\lfloor k/4\rfloor$ consecutive positions must
admit a \kloop. This follows from pumping arguments.

\begin{lem} \label{lem:cloop} Let $w$ be a word and let
  $x_1,\ldots,x_{\lfloor k/4\rfloor}$ be $\lfloor k/4\rfloor$
  consecutive positions in~$w$. Then, there exists at least one
  position $x_i$ with $i < \lfloor k/4\rfloor$  that admits a \kloop.
\end{lem}

\begin{proof}
  By choice of $k$, $\lfloor k/4\rfloor=|M|+1$. By the
  pigeonhole principle, we obtain two natural numbers $1 \leq i < j \leq \lfloor k/4\rfloor$ 
  such that
  $\alpha(w[x_1,x_i])=\alpha(w[x_1,x_j])$. We prove that $x_i$
  admits a \kloop. Consider the \profile{ \lfloor k/2\rfloor}
  $(w_\ell,w_r)$ of $x_i$ and $u=w[x_i,x_j]$. Since
  $|w[x_1,x_i]|<\lfloor k/4\rfloor$, $w[x_1,x_i]$ is a suffix of
  $w_\ell$, so
  $\alpha(w_\ell) = \alpha(w_\ell \cdot u)$. Since $|u| < \lfloor
  k/4\rfloor$, $u$ is a prefix of $w_r$. Therefore $x_i$ admits a \kloop.

  Observe that $u$ is not necessarily \emph{the} \kloop of $x_i$, as there
  might be a smaller word that also satisfies the definition.
\end{proof}

\medskip
\noindent
{\bf Construction of $w'_1,w'_2$.} We can now construct $w'_1$ and
$w'_2$. If $w,u$ are words and $x$ is a position of $w$, the
\emph{word constructed by inserting $u$ at position $x$} is the word
$w[0,x] \cdot u \cdot w[x,|w|]$. From $w_1$ (resp.~$w_2$), we
construct $w'_1$ (resp.~$w'_2$) by inserting simultaneously all
infixes $z_x$ in $w_1$ (resp.~$w_2$) at any position $x$ that admits a
\kloop, and where $z_x$ is the \kloop of $x$. It remains to prove that
the construction is correct, \emph{i.e.}, that for $i=1,2$, $w'_i$ admits an
$(\alpha,\alpha(w_i))$-compatible \pfsdecomp for some $d$-pattern
$\Ps=(\frp,f,\frs)$.

\medskip
\noindent
{\bf The Construction is Correct.} If we had
$\min(|w_1|,|w_2|)\le\lfloor k/4\rfloor$, then from $w_1 \ltteq{k}{d} 
w_2$ one would obtain $w_1=w_2$, a case already
excluded. Hence, both $w_1,w_2$ have length at least $\lfloor
k/4\rfloor$, so by Lemma~\ref{lem:cloop}, at least one insertion has
occurred in both $w_1$ and $w_2$.  We now prove that
there exists a $d$-pattern $\Ps=(\frp,f,\frs)$ such that for $i=1,2$,
$w'_i$ admits an $(\alpha,\alpha(w_i))$-compatible \pfsdecomp. 

We first define the $d$-pattern $(\frp,f,\frs)$. By definition, $w'_1$
can be decomposed as $w'_1 = u_0v_1u_1v_2 \cdots v_{n}u_n$ with
$w_1=u_0u_1\cdots u_n$ and the words $v_j$ are the \kloops inserted in
the construction. Since at least one insertion was made, we have $n
\geq 1$ and we can set $\frp = (u_0,v_1)$, $\frs = (v_n,u_n)$.  We
define $f$ as the function that maps a block $(v_\ell,u,v_r)$ to the
number of times it occurs in the decomposition, up to
threshold~$d$. Set $\Ps=(\frp,f,\frs)$.  By definition, $u_0v_1u_1v_2
\cdots v_{n}u_n$ is a \pfsdecomp for $w'_1$. Moreover, it is
$(\alpha,\alpha(w_1))$-compatible by Fact~\ref{fct:tloop}. It
remains to prove that $w'_2$ admits an
$(\alpha,\alpha(w_2))$-compatible \pfsdecomp.

By definition, $w'_2$ can be decomposed in a similar way as $w'_1$, 
$w'_2 = u'_0v'_1u'_1v'_2 \cdots v'_{m}u'_m$ with $w_2=u'_0u'_1\cdots
u'_m$ and the words $v'_j$ are the \kloops inserted in the
construction of~$w'_2$. We prove that this is a \pfsdecomp. It will then be
$(\alpha,\alpha(w_2))$-compatible by Fact~\ref{fct:tloop}.

To every position $x$ in $w_2$ we associate a block, prefix block or
suffix block in the following way. By definition, $x$ must fall into a
word $u'_i$ for some $i$. If $i \notin \{0,m\}$, we denote by $g(x)$
the triple $(v'_i,u'_i,v'_{i+1})$. Similarly if $i=0$ (resp.~$i=m$),
then $g(x)$ is the pair $(u'_0,v'_{1})$ (resp.~$(v'_m,u'_{m})$). The
result now follows from the next lemma.

\begin{lem} \label{lem:blockscst} Let $x,y$ be distinct positions of
  words $w_1$ or $w_2$ with the same \profile{k}. Then
  $g(x)=g(y)$. Moreover, the number of copies of a block \frb in the
  decomposition of $w_1$ (resp. $w_2$) is exactly the number of
  positions $x$ in $w_1$ (resp. $w_2$) such that $g(x)=\frb$.
\end{lem}

\begin{proof}
  There are several cases to treat depending on whether $x,y$ are in
  $w_1$ or $w_2$ and whether there are at least $\lfloor k/4\rfloor$
  positions to the right and left of $x,y$ or not. All cases are
  treated similarly. Therefore, we only treat the case when there are
  at least $\lfloor k/4\rfloor$ positions to the right and left of
  $x,y$, and positions $x,y$ are both in $w_1$.

  Let $z$ be a position such that there are at least $\lfloor k/4
  \rfloor$ positions to the right and left of $z$. Let $(w_\ell,w_r)$
  be the \kprofile of $z$ and $z_\ell \leq z$ and $z_r > z$ be the
  positions admitting \kloops that are closest to $z$. Observe that by
  Lemma~\ref{lem:cloop}, $z-z_\ell \leq {\lfloor k/4\rfloor}$ and $z_r-z
  \leq {\lfloor k/4\rfloor} - 1$, hence $z_\ell,z_r$ belong to the copy of $w_\ell,w_r$
  at $z$.

  We claim that  the relative positions of $z_\ell,z_r$ in this copy and the
  actual \kloops only depend on $(w_\ell,w_r)$. This claims immediately implies,
  by construction of $w'_1,w'_2$, that $g(z)$ only depends on its
  \kprofile. This entails the first part of the lemma. Moreover,
  this also proves that if $(v_\ell,u,v_r) = g(z)$, then the relative
  position of $z$ within the corresponding copy of $u$ only depends on
  the \kprofile of $z$. This means that two positions with the same
  \kprofile can only generate the same copy of a block if they are
  equal: this is the second part of the~lemma.

  It remains to prove the claim. By definition of profiles, the
  \profiles{\lfloor k/2\rfloor} of $z_\ell,z_r$ are determined
  by the \kprofile of $z$ (see Figure~\ref{fig:profile}).
  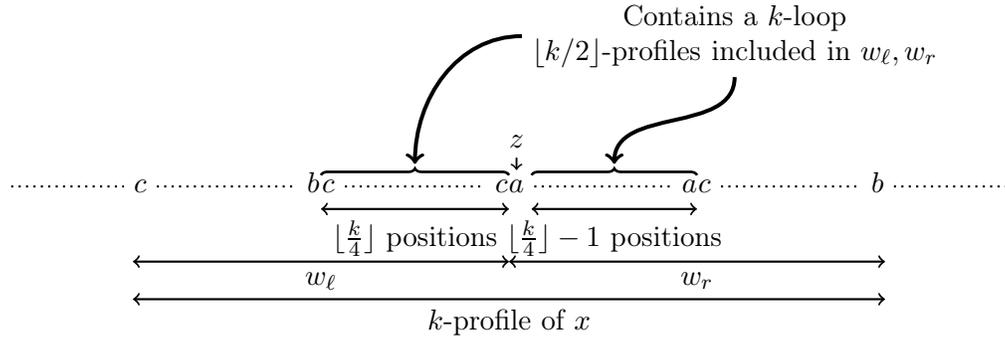
\begin{figure}[h]
    \begin{center}
      \begin{tikzpicture}
        \node[anchor=mid] (b1) at (-4.9,0.0) {$c$};
        \node[anchor=mid] (b2) at (-2.6,0.0) {$b$};
        \node[anchor=mid] (b3) at (-2.4,0.0) {$c$};
        \node[anchor=mid] (b4) at (-0.1,0.0) {$c$};

        \draw[dotted,thick] (b1.mid east) to (b2.mid west);
        \draw[dotted,thick] (b3.mid east) to (b4.mid west);
        \draw[dotted,thick] ($(b1.mid west)-(1.5,0.0)$) to (b1.mid west);

        \node[anchor=mid] (a1) at (0.1,0.0) {$a$};
        \node[anchor=mid] (a2) at (2.4,0.0) {$a$};
        \node[anchor=mid] (a3) at (2.6,0.0) {$c$};
        \node[anchor=mid] (a4) at (4.9,0.0) {$b$};

        \draw[dotted,thick] (a1.mid east) to (a2.mid west);
        \draw[dotted,thick] (a3.mid east) to (a4.mid west);
        \draw[dotted,thick] (a4.mid east) to ($(a4.mid east)+(1.5,0.0)$);

        \node[anchor=mid] (z1) at (0.1,0.6) {$z$};

        \draw[ars] (z1) to (a1);
        \draw[<->,thick] (-2.5,-0.3) to node[below] {${\lfloor \frac{k}{4}\rfloor}$ positions}
        (0.0,-0.3);
        \draw[<->,thick] (0.3,-0.3) to node[below] {${\lfloor \frac{k}{4}\rfloor}-1$ positions}
        (2.5,-0.3);
        \draw[acc] (-2.5,0.15) to (0.0,0.15);
        \draw[acc] (0.3,0.15) to (2.5,0.15);
        \draw[<->,thick] (-5.0,-1.0) to node[below] {$w_\ell$} (0.0,-1.0);
        \draw[<->,thick] (0.0,-1.0) to node[below] {$w_r$} (5.0,-1.0);
        \draw[<->,thick] (-5.0,-1.5) to node[below] {\profile{k} of $x$} (5.0,-1.5);

        \node[align=center] (info) at (3.0,2.0) {Contains a
          \kloop\\\profiles{\lfloor k/2\rfloor} included in $w_\ell,w_r$};

        \draw[arp] (info.west) to [out=180,in=90] (-1.25,0.3);
        \draw[arp] (info.south) to [out=-90,in=90] (+1.4,0.3);
      \end{tikzpicture}
    \end{center}
    \caption{Construction in Lemma~\ref{lem:blockscst}}
    \label{fig:profile}
  \end{figure}
  By Fact~\ref{fct:profloop}, this means that $z_\ell,z_r$ as well as
  their actual \kloop are determined by the \kprofile of $z$, which
  terminates the proof.
\end{proof}

Taking $x,y$ as the first positions (resp.~last positions) of
$w_1,w_2$, Lemma~\ref{lem:blockscst} implies that $\frp = (u_0,v_1) =
(u'_0,v'_1)$ (resp.~$\frs = (v_n,u_n) = (v'_m,u'_m)$). We finish by
proving that for every block $(v_\ell,u,v_r)$ the number of indices
$i$ such that $(v_\ell,u,v_r)=(v_i,u_i,v_{i+1})$ is the same in the
decompositions of $w'_1,w'_2$ up to threshold $d$. By definition of
$(\frp,f,\frs)$ this will prove that the decomposition of $w'_2$ is
indeed a \pfsdecomp.

Let $\frb=(v_\ell,u,v_r)$ be a block. By Lemma~\ref{lem:blockscst},
there exists a set $P$ of \kprofiles such that the number of indices $i$ in
$w' _1$ (resp.~$w'_2$) such that $\frb = (v_i,u_i,v_{i+1})$ is exactly
the number of positions in $w_1$ (resp.~in $w_2$) with a \kprofile in
$P$. Because $w_1 \kdltteq w_2$ these numbers are then equal in $w_1,
w_2$ up to threshold $d$ and this finishes the proof of Proposition~\ref{prop:ktodecomp}.

\section{Separation by \ltt Languages}
\label{sec:ltt}
\makeatletter{}%
This section is devoted to \ltt. Again, our theorem actually contains
several results. In the case of \ltt, two parameters are involved: the
size $k$ of profiles and the counting threshold~$d$. The first result
in our theorem states that the bound on $k$ of Theorem~\ref{thm:seplttd}
still holds for full \ltt. This means that two languages are
\ltt-separable if and only if there exists some counting threshold $d$
such that they are $\ltt[k,d]$-separable with the same bound $k$ as in
Theorem~\ref{thm:seplttd}. It turns out that this already yields an
algorithm for testing \ltt-separability. The algorithm relies on the
decidability of Presburger arithmetic and is actually adapted in a
straightforward manner from an algorithm of~\cite{bojLTT} for deciding
membership in \ltt.

While this first result gives an algorithm for testing separability,
it gives no insight about an actual separator. Indeed, the
procedure does not produce the actual counting threshold~$d$. This is
obtained in the second part of our theorem: we prove that two
languages are \ltt-separable if and only if they are
$\ltt[k,d]$-separable, where~$k$ is as defined in
Theorem~\ref{thm:seplttd}, and $d$ is bounded by 
a function of the size of the monoid (or automaton) recognizing the
input languages. Note that this result also gives another
(brute-force) algorithm for testing \ltt-separability.
We now state our theorem. Recall that $A_k$ denotes
the set of \kprofiles.

\begin{thm}
  \label{thm:sepltt}
  Let $L_1,L_2$ be regular languages. Let $\alpha:A^*\to
  M$ be a morphism into a finite monoid $M$ recognizing both $L_1$ and $L_2$. Let $\As$
  be an NFA recognizing both $L_1$ and $L_2$, such that $L_i=L(\As,I_i,F_i)$.  Set $n$ to
  be either $|M|+1$ or $|\As|+1$. Let $k=4(|M|+1)$ and let $d=(|A_k|n)^{|A_k|}$. Then, the following
  conditions are equivalent:
  \begin{enumerate}
  \item\label{item:a3} $L_1$ and $L_2$ are \ltt-separable.
  \item\label{item:b3} There exists $d' \in \nat$ such that $L_1$ and
    $L_2$ are $\ltt[k,d']$-separable.
  \item\label{item:e3} There exists $d' \in \nat$ such that no pair in
    $\alpha(L_1)\times\alpha(L_2)$ has a common $d'$-pattern.
  \item\label{item:f3} There exists $d' \in \nat$ such that no pair in
    $(I_1\times F_1) \times(I_2\times F_2)$ has a common $d'$-pattern.
  \item\label{item:c3} $L_1$ and $L_2$ are $\ltt[k,d]$-separable.
  \item\label{item:d3} The language $\lttclos{L_1}{k}{d}$ separates
    $L_1$ from $L_2$.
  \end{enumerate}
\end{thm}

Observe that decidability of \ltt-separability is immediate from
Item~\eqref{item:c3} by using the usual brute-force algorithm. As it was
the case for a fixed counting threshold, this algorithm is
slow and we will present a faster algorithm by using
Items~\eqref{item:e3} and~\eqref{item:f3} in Section~\ref{sec:comp}.

\begin{cor}
  \label{cor:decidltt}
  It is decidable whether two given regular languages are
  \ltt-separable.
\end{cor}

By definition, a language is \ltt if it is $\ltt[k,d]$ for
some natural numbers $k,d$. Hence, the equivalence between
Items~\eqref{item:a3}, \eqref{item:b3}, \eqref{item:e3} and~\eqref{item:f3} is 
an immediate consequence of Theorem~\ref{thm:seplttd}. It is also
clear that
$\eqref{item:c3} \implies\eqref{item:d3}\implies\eqref{item:a3}$.
Therefore, we
only need to prove Items~\eqref{item:c3} or \eqref{item:d3} from one
of the other properties, \emph{i.e.}, the
bound on the threshold $d$. Unfortunately, these are exactly the items
we need for Corollary~\ref{cor:decidltt}. However, we will prove that
by reusing an algorithm of~\cite{bojLTT}, Corollary~\ref{cor:decidltt}
can also be derived directly from Item~\eqref{item:b3}.

The remainder of this section is organized in three
subsections. We first explain how Corollary~\ref{cor:decidltt} can be
derived from Item~\eqref{item:b3} without actually having to compute a
bound on the counting threshold. Next, we prove
our bound on the counting threshold in Theorem~\ref{thm:sepltt}. Finally,
we discuss the optimality of this bound in the last subsection.

\subsection{Decidability of \ltt-separability as a consequence of
  Theorem~\ref{thm:seplttd}}

As we explained, the equivalence of Item~\eqref{item:b3} to
\ltt-separability is immediate from Theorem~\ref{thm:seplttd}. 
We explain how to combine Item~\eqref{item:b3} with an algorithm
of~\cite{bojLTT} to obtain decidability directly.

In~\cite{bojLTT}, it is proved that once $k$ is fixed, Parikh's
Theorem~\cite{Parikh:Context-Free-Languages:1966:a} can be used to
prove that whether a language is $\ltt[k,d]$ for some $d$ can be
rephrased as a computable Presburger formula. Decidability of
membership in \ltt can then be reduced to decidability of Presburger
Arithmetic. For achieving this, two ingredients were needed: $a)$ a
bound on $k$, and $b)$ the translation to Presburger arithmetic. It
turns out that in~\cite{bojLTT}, only the proof of $a)$ was specific to
membership. On the other hand, separation was already taken care of in
$b)$, because the intuition behind the Presburger formula was 
testing \emph{separability} between the input language and its complement. In
our setting, we have already replaced~$a)$, \emph{i.e.},~bounding~$k$,
by Item~\eqref{item:b3}. Therefore, the argument can be generalized. We
explain in the remainder of this subsection how to construct the Presburger
formula. The argument makes use of the notion of
commutative image, which we first recall.

\medskip
\noindent
{\textbf{Commutative Images.}} Let $w\in A^{*}$. The \emph{commutative
  image} of $w$, denoted by $\pi(w)$, is the $A$-indexed vector of
natural numbers counting, for every $a \in A$, how many occurrences of $a$
there are~in~$w$. This notion can be easily generalized in order to
count profiles rather than just letters. Let $k\in\nat$. The
\emph{\kimage } of $w$, $\pi_k(w)$, is the $A_k$-indexed vector of
numbers  counting for every \kprofile $(w_\ell,w_r)$ the number of
positions in $w$ with \kprofile $(w_\ell,w_r)$.  If $L$ is a language, the
\emph{\kimage } of $L$, $\pi_k(L)$ is the set $\{\pi_k(w) \mid w \in
L\}$. The definition of \kdltteq yields the following~fact.

\begin{fact} \label{lem:eqredef}
  Let $w,w' \in A^{*}$ and let $k,d \in \nat$. Then $w \kdltteq w'$ if and only if
  $\pi_k(w)$ and $\pi_k(w')$ are equal componentwise up to threshold $d$.
\end{fact}

A well-known result about commutative images is Parikh's
Theorem~\cite{Parikh:Context-Free-Languages:1966:a}, which states that
if $L$ is context-free (and so in particular if $L$ is
regular), then $\pi(L)$ is semilinear, \emph{i.e.}, Presburger
definable~\cite{Ginsburg&Spanier:Semigroups-Presburger-Formulas-Languages:1966:a}. As
explained in~\cite{bojLTT}, Parikh's Theorem extends without
difficulty to \kimages.

\begin{thm} \label{prop:parikh}
  Let $L$ be a context-free language and let $k \in \nat$. Then $\pi_k(L)$
  is semilinear. Moreover, a Presburger formula for this
  semilinear set can be computed from $L$. 
\end{thm}

\begin{proof}
  When $k=1$, the proposition is Parikh's Theorem. When $k > 1$,
  consider the language $L'$ over the alphabet $A_k$ of \kprofiles
  such that $w' \in L'$ if and only if there exists $w \in L$ of the
  same length and such that a position in $w'$ is labeled by the
  \kprofile of the same position in $w$. It is straightforward to see that $L'$ is
  context-free and that the \kimage  $\pi_k(L)$ of $L$ is its
  commutative image $\pi(L')$, which
  is semilinear by Parikh's Theorem.
\end{proof}

We can now explain how to decide \ltt-separability. By
Item~\eqref{item:b3} in Theorem~\ref{thm:sepltt}, $L_1,L_2$ are 
\ltt-separable if and only if they are $\ltt[k,d]$-separable for
$k=4(|M|+1)$ (where $M$ is a monoid recognizing both $L_1,  
L_2$) and some natural number $d$. Therefore, whether $L_1,L_2$ are
\ltt-separable can be rephrased as follows: does there exist some
threshold $d$ such that there exist no words $w_1 \in L_1,w_2 \in L_2$
such that $w_1 \kdltteq w_2$? By Fact~\ref{lem:eqredef}, this can be
expressed in terms of \kimages: does there exist a threshold $d$ such
that there exist no vectors of natural numbers
$\bar{x}_1\in\pi_k(L_1),\bar{x}_2\in\pi_k(L_2)$ that are 
equal up to threshold $d$? It follows from Theorem~\ref{prop:parikh}
that the above question can be expressed as a computable Presburger
formula. Decidability of  \ltt-separability then follows from
decidability of Presburger Arithmetic.

\subsection{Bounding the Counting Threshold}
\makeatletter{}%
We now prove the bound on the counting threshold~$d$. This amounts to proving
that Items~\eqref{item:c3} and~\eqref{item:d3} in Theorem~\ref{thm:sepltt}
are equivalent to Item~\eqref{item:b3}. By definition, it is immediate that
$\eqref{item:d3}\Rightarrow\eqref{item:c3}\Rightarrow \eqref{item:b3}$.
It remains to prove $\eqref{item:b3}\Rightarrow\eqref{item:d3}$.

\medskip
We actually prove the contrapositive: if $\lttclos{L_1}{k}{d}$ does
not separate $L_1$ from $L_2$ for the values of $k,d$ defined in the theorem,
then there is no $\ltt[k,d']$-separator for any $d'$.

\smallskip
Assume that $\lttclos{L_1}{k}{d}$ is not a separator. By definition,
this means that there exist $w_1 \in L_1$ and $w_2 \in L_2$ such that
$w_1 \kdltteq w_2$. Set $d'$ some arbitrary natural number. We prove that $d$
is large enough to construct words $w'_1 \in L_1,w'_2 \in L_2$ such
that $w'_1 \ltteq{k}{d'} w'_2$. By definition of $\ltteq{k}{d'}$, this
means that there exists no $\ltt[k,d']$-separator, which is what we
want to prove.

\smallskip
Recall that $n=|M| + 1$ or $n=|\As|+1$. To simplify the writing, we only treat the case when
$n=|\As|+1$.  The other case can be proved similarly. 
Set $m=|A_k|n$, so that $d=m^{|A_k|}$. For $\ell \leq |A_k|$, we 
prove the following property by induction:
\begin{equation*}
  \Ps(\ell)\qquad
  \begin{aligned}
    &\text{If $u_1\in L_1$, $u_2 \in L_2$ are such that $u_1
      \ltteq{k}{m^{\ell}} u_2$, then for all $d'$, if the number of}\\[-1ex]
    &\text{\kprofiles that do not
      occur more than $d'$ times in both $u_1$ and $u_2$ is smaller}\\[-1ex]
    &\text{than $\ell$, then there exist words $u'_1 \in L_1$ and $u'_2 \in
      L_2$ such that $u'_1 \ltteq{k}{d'} u'_2$.}
  \end{aligned}
\end{equation*}

Before proving $\Ps(\ell)$, note that by definition of $d=m^{|A_k|}$ and since $A_k$ is the set of
\emph{all} \kprofiles, $w_1,w_2$ verify the premise of
$\Ps(|A_k|)$. Therefore, $ \Ps(|A_k|)$ entails that there exist words
$w'_1,w'_2$ such that the desired property $w'_1 \ltteq{k}{d'} w'_2$
holds for all $d'$. It remains to prove $\Ps(\ell)$ for $\ell \leq |A_k|$, which we do by
induction on $\ell$.

The main idea for the inductive step is that by choice of $d$,
\kprofiles that occur many times in $w_1,w_2$ can be pumped in more
than $d'$ occurrences. This is summarized in the following lemma, that
we prove below.

\begin{lem}
  \label{lem:thresholdpump}
  Let $d'\in\nat$. Let $h\ge1$ be a natural number, $h' \geq n|A_k|h$ and $w
  \in L_1$ (resp.~$w \in L_2$). One can construct a word $w' \in L_1$
  (resp.~$w' \in L_2$) such that $w' \ltteq{k}{h} w$ and every
  \kprofile that occurs $h'$ or more times in $w$ occurs $d'$ or more
  times in $w'$.
\end{lem}

Assume for a moment that Lemma~\ref{lem:thresholdpump} holds and let
us deduce that $\Ps(\ell)$ holds, by induction on $\ell$.
If $\ell=0$, the result is obvious since by definition all \kprofiles
in $u_1,u_2$ occur more than $d'$ times in both words and therefore
$u_1 \ltteq{k}{d'} u_2$.

Assume now that $\ell > 0$. If $u_1 \ltteq{k}{d'} u_2$, then it suffices
to take $u'_1=u_1$ and $u'_2=u_2$. Otherwise, since $u_1
\ltteq{k}{m^{\ell}} u_2$, this means that there exists at least
one \kprofile $(w_\ell,w_r)$ that occurs more than $m^{\ell}$
times in both $u_1$ and $u_2$ but less than $d'$ times in at least one
of the two words. By applying Lemma~\ref{lem:thresholdpump} to both
$u_1$ and $u_2$ with $h=m^{\ell-1}$ and $h'=m^{\ell}$, we get
two words $u''_1\in L_1$ and $u''_2 \in L_2$ such that $u''_1
\ltteq{k}{h} u''_2$. Moreover, $(w_\ell,w_r)$ now occurs more than
$d'$ times in both $u''_1$ and $u''_2$. Therefore, the number of
\kprofiles that do not occur more than $d'$ times in both $u''_1$ and
$u''_2$ is smaller than $\ell-1$. Hence, we can apply the induction
hypothesis to $u''_1$ and $u''_2$ which yields the desired
$u'_1,u'_2$.

\medskip\noindent
To conclude the proof, it remains to show
Lemma~\ref{lem:thresholdpump}.

\begin{proof}[Proof of Lemma~\ref{lem:thresholdpump}]
  By symmetry, we may assume that $w \in L_1$.
  The intuition is that as soon as a \kprofile occurs more than
  $|\As|+1$ times (which is the case if it occurs more than $n$ times),
  there exist two occurrences of this \kprofile that are labeled with
  the same state in the run of $\As$ on $w$. Therefore, pumping
  can be used on $w$ to generate $d'$ copies of the \kprofile
  without affecting membership in $L_1$. The issue with this argument
  is that in order to enforce $w' \ltteq{k}{h} w$, we need to be
  careful and avoid duplicating \kprofiles that occur less than $h$
  times in $w$. This is why we actually need a much higher constant
  than $n$ in order to achieve the~pumping.

  Let $(w_\ell,w_r)$ be some \kprofile that occurs more than $h'$
  times in $w$ (if there is none, it suffices to take $w'=w$). We
  explain how $w$ can be pumped in $w'$ that contains more than $d'$
  copies of $(w_\ell,w_r)$ while enforcing $w' \ltteq{k}{h} w$. The
  construction can then be repeated to treat all \kprofiles occurring
  more than $h'$ times in $w$, in order to get the desired $w'$.

  Let $x_1 < \dots < x_{h'}$ be $h'$ positions where $(w_\ell,w_r)$
  occurs. Observe that there are at most $|A_k|(h-1)$ positions in $w$
  such that the \kprofile at this position occurs strictly less than
  $h$ times in $w$. By choice of $h'\geq n|A_k|h$, a simple application of the
  pigeonhole principle yields that there exist at least $n$
  consecutive positions in the list, say $x_i,\dots,x_{i+(n-1)}$, such
  that no intermediate position between $x_i$ and $x_{i+(n-1)}$ has a
  \kprofile occurring less than $h$ times in $w$. By choice of $n$,
  there are two positions among $x_i, \dots, x_{i+(n-1)}$ that are
  labeled with the same state in the run of $\As$ on
  $w$. Therefore, the corresponding infix can be pumped to generate
  $d'$ copies of $(w_\ell,w_r)$ without affecting membership in
  $L_1$. Moreover, by choice of the positions $x_i,\dots,x_{i+(n-1)}$,
  the pumping did not duplicate \kprofiles occurring less than $h$
  times in $w$. Therefore, the resulting word $w'$ verifies
  $w'\ltteq{k}{h} w$ and $w' \in L_1$.
\end{proof}

\subsection{Optimality of the Counting Threshold}
\label{sec:append-sect-ltt}

Observe that the bound of Theorem~\ref{thm:sepltt} for the
counting threshold is exponential in the size of $|A_k|$ (which is
itself exponential in the size of the monoid). A relevant question is
to know whether this bound can be improved.

Our proof completely separates the bounding of $k$ and $d$. We first
provide a bound on $k$ in Theorem~\ref{thm:seplttd}. Then, we bound
the threshold by essentially viewing our words as words over the
alphabet $A_k$ of \kprofiles. This technique ignores properties of
\kprofiles. In particular, the \kprofiles of adjacent positions are
strongly related, a fact that our proof does not exploit.

In this subsection, we prove that getting a better bound on the
counting threshold would require taking these additional properties
into account. More precisely, we prove that if $k=1$ (which means that the
\kprofile of a position is just its label), we can construct separable
languages for which the separator requires a counting threshold that
is exponential in $|A|$.

For convenience, we assume the alphabet $A$ to be of even size and
write $A=\{a_1,\dots,a_{2m}\}$. Consider the following languages
\[
\begin{array}{lcl}
  L_1 & = & a_1\cdot(a_2a_3a_3)^*(a_4a_5a_5)^* \cdots (a_{2m-2}a_{2m-1}a_{2m-1})^*\cdot(a_{2m}a_{2m}a_{2m})^*,
  \\
  L_2 & = & (a_1a_2a_2)^*(a_3a_4a_4)^* \cdots (a_{2m-1}a_{2m}a_{2m})^*.
\end{array}
\]

\begin{lem} \label{lem:counterexample}
  $L_1,L_2$ are $\ltt[1,d]$-separable for some $d$, but not 
  $\ltt[1,2^{2m-1}]$-separable.
\end{lem}

\begin{proof}
  We prove that $L_1,L_2$ are $\ltt[1,d]$-separable for $d =
  2^{2m-1}+1$. Consider the following language $L$. A word $w$ belongs
  to $L$ if and only if for all odd $i$ either $w$ contains at least
  $2^{i-1}+1$ copies of $a_i$, or the number of copies of $a_{i+1}$ in
  $w$ is exactly twice the number of copies of $a_i$ in $w$.

  The maximal threshold up to which we have to count occurrences of
  letters for checking membership in $L$ of a word $w$ is the case
  where $w$ has exactly $2^{2m-2}$ occurrences of $a_{2m-1}$, and
  therefore exactly $2^{2m-1}$ occurrences of $a_{2m}$. This shows
  that $L \in \ltt[1,2^{2m-1}+1]$.

  We prove that $L$ is a separator. By definition, we have $L
  \supseteq L_2$. We prove that $L \cap L_1 = \varnothing$. By
  contradiction, assume that $w \in L \cap L_1$. Since $w \in L_1$, it
  contains only one copy of $a_1$. Then, since $w \in L$, it must
  contain two copies of $a_2$. By iterating the argument we get that
  $w$ must contain $2^{2m-1}$ copies of $a_{2m}$, which is impossible
  since this number must be multiple of $3$, by definition of $L_1$. It
  follows that $L \in \ltt[1,2^{2m-1}+1]$ is a separator.

  It remains to prove that $L_1,L_2$ are not
  $\ltt[1,2^{2m-1}]$-separable. Consider the words 
  \[
  \begin{array}{lcl}
    w_1 = a^{}_1(a_2a_3^2)^2 \cdots (a_{2m-2}^{}a_{2m-1}^2)^{2^{2m-3}} (a_{2m})^{3
      \times 2^{2m-1}} & \in & L_1, \\
    w_2 = (a^{}_1a_2^2)(a^{}_3a_4^2)^4 \cdots (a^{}_{2m-1}a_{2m}^2)^{2^{2m-2}}  & \in & L_2.
  \end{array}
  \]
  It is clear that $w_1 \ltteq{1}{2^{2m-1}} w_2$. Therefore $L_1,L_2$
  are not $\ltt[1,2^{2m-1}]$-separable.
\end{proof}

\section{Complexity: upper and lower bounds}
\label{sec:comp}
\makeatletter{}%
In this section, we present lower and upper complexity bounds for
deciding \lt- and \ltt-separability.

Both the lower and upper bounds rely on the pattern criteria of
Theorems~\ref{thm:seplt} and~\ref{thm:sepltt}. We are able to prove
that starting from NFAs or DFAs recognizing the input languages,
deciding separability can be achieved in \textsc{co-Nexptime} for \lt
and in \textsc{2-Expspace} for \ltt. Moreover, we prove a
\textsc{co-Np} lower bound for both problems. 

\subsection{Upper Bounds}

\makeatletter{}%
In this subsection, we prove the complexity upper bounds %
for the
separation problem for both \lt and
\ltt. The corresponding algorithms rely on the patterns criteria of
Theorems~\ref{thm:seplt} and~\ref{thm:sepltt}. We prove the two
following results.

\begin{prop} \label{prop:ult}
 Deciding whether two languages accepted by some NFA
 are \lt-separable can be achieved in {\sc co-Nexptime}.
\end{prop}

\begin{prop} \label{prop:ultt}
  Deciding whether two languages accepted by some NFA
  are \ltt-separable can be achieved in {\sc 2-Expspace}.
\end{prop}

Both algorithms work by reducing the problems to the special case
of $k=1$, \emph{i.e.}, whether there exists an \lt (resp.~\ltt) separator
that considers only \profiles{1}. The reduction is identical in both
cases and the difference resides in proving that it is correct. These
proofs rely on Item~\eqref{item:e2} in Theorem~\ref{thm:seplt} and
Item~\eqref{item:f3} in Theorem~\ref{thm:sepltt} respectively. The
computations involved in the reduction can be done in {\sc Exptime}
and the new NFA it outputs is of size exponential in the input NFA.
Therefore, it then suffices to give algorithms for the special case
$k=1$, which run in {\sc co-Np} for \lt and {\sc Expspace} for \ltt.

Note that a reduction to $k=1$ could also be done by considering
the bound $k$ on the size of profiles in Theorems~\ref{thm:seplt}
and~\ref{thm:sepltt}. Indeed, once $k$ is fixed it suffices to modify
the input NFA to work on the alphabet of \kprofiles to be reduced to
the case $k=1$. However, this technique might yield an NFA that
is doubly exponential in the size of the input NFA.

We first present and prove the reduction to the case $k=1$. Then we
explain how to decide both problems in this special case.

\subsubsection{Reduction to the case $k=1$}
\label{sec:reduction-case-k=1}

Let $\As = (Q,A,\delta)$ be an NFA. It follows from Theorem~\ref{thm:seplt}
(resp.~Theorem~\ref{thm:sepltt}) that to determine whether $L(\As,I_1,F_1)$
and $L(\As,I_2,F_2)$ are not \lt-separable (resp.~not \ltt-separable), it
suffices to verify whether there exists a pair in $(I_1\times F_1)
\times(I_2\times F_2)$ having a common $1$-pattern (resp.~whether there exists
$d$ such that some pair in $(I_1\times F_1) \times(I_2\times F_2)$ has a
common $d$-pattern). This requires verifying whether there exist a pattern \Ps
and \pfsdecomps compatible with such a pair.

Observe that a \pfsdecomp can be viewed as a word over the alphabet of
blocks. The main idea behind the reduction is to construct a new NFA
$\widetilde{\As}$ recognizing words over the
alphabet of blocks that represent \pfsdecomps that, for some \Ps, are compatible
with pairs of states of $I_1\times F_1$ and of $I_2\times F_2$ in
$\As$. There are two main issues with
this outline. First, the alphabet of blocks is infinite. Second, a
notion of compatibility has to be enforced between consecutive blocks
in a word, \emph{i.e.}, $(v,u,v')$ needs to be followed by $(v',u',v'')$ for
some $u',v''$. Again, this compatibility cannot be simply encoded in
the states since there are infinitely many words.

Both issues are solved using a similar argument. Recall that we are
only interested in \pfsdecomps that are compatible with pairs of states of $I_1\times F_1$  and of $I_2\times F_2$ in~$\As$. Observe that for a block $(v,u,v')$ to appear in such a
decomposition, there need to exist states $q,q'$ with loops at $q,q'$
labeled by $v,v'$ respectively and a path from $q$ to $q'$ labeled by
$u$. We abstract a block by the set of pairs of states verifying this
property. Since there are finitely many such sets, this yields a
finite alphabet. The same argument can be used for compatibility: all
words $v$ that need to be considered for compatibility need to label a
loop at some state $q$. We abstract $v$ by the set of states having
such a loop labeled~$v$. The information can then be encoded in the
states of $\widetilde{\As}$.

\smallskip We now describe the formal construction. Let $R \subseteq
Q$. We say that $R$ is \emph{synchronizable} if there
exists a nonempty word $v \in A^*$ such that for all $q \in R$ there exists a loop
at $q$ labeled by $v$. We will encode compatibility in
synchronizable sets of states.

If $T \subseteq Q^2$, we denote by $\ell(T)$ (resp.\ by~$r(T)$) the set of
states that are left (resp.~right) members of pairs in $T$. We say that $T$ is
\emph{synchronizable} if {\bf a)} there exists $u \in A^*$ such that for all
$(q,q') \in T$ there exists a run from $q$ to $q'$ labeled by $u$, {\bf b)}
$\ell(T)$ is synchronizable and {\bf c)} $r(T)$ is synchronizable. In order to
deal with prefixes and suffixes, we generalize the notion to these limit
cases. We say that $T$ is \emph{prefix synchronizable} (resp.~\emph{suffix
  synchronizable}) if {\bf a)} and {\bf c)} (resp.~{\bf a)} and {\bf b)})
hold. Finally, we say that $T$ is \emph{weakly synchronizable} if %
{\bf a)} holds. We set $B_w,B_p,B_i$ and $B_s$ as the sets of weakly
synchronizable, prefix synchronizable, synchronizable and suffix
synchronizable sets of pairs of states, respectively.

We can now define $\widetilde{\As}$. Set the alphabet $B$ as the disjoint
union of $B_w,B_p,B_i$ and $B_s$. The state set $\widetilde{Q}$ is defined as follows.
\[
\widetilde{Q} = \{(r,R)\ |\ r \in R \subseteq Q,\ R \text{ is
  synchronizable}, r \in Q\} \cup I_1 \cup F_1\cup I_2 \cup F_2.
\] 
We now define the transitions. For all $b \in B_w$, we add a $b$-transition
from $q_k \in I_k$ to $r_k \in F_k$ if $(q_k,r_k) \in b$. For all $b \in B_p$,
we add a $b$-transition from $q_k \in I_k$ to $(r_k,R_k)$ if $(q_k,r_k) \in b$
and $R_k=r(b)$. Similarly, for all $b \in B_s$, we add a $b$-transition from
$(r_k,R_k)$ to $q_k \in F_k$ if $(r_k,q_k) \in b$ and $R_k =
\ell(b)$. Finally, for all $b \in B_i$, we add a $b$-transition from
$(r_k,R_k)$ to $(s_k,S_k)$ if $(r_k,s_k) \in b$, $R_k = \ell(b)$
and~$S_k=r(b)$.

\smallskip
This ends the construction of $\widetilde{\As}$ from $\As,I_1,F_1,I_2,F_2$. Observe that it is of size
exponential in the size of $\As$. We now prove that the
computation can be done in {\sc Exptime}.

\begin{lem} \label{lem:construc}
  Given an NFA  $\As,I_1,F_1,I_2,F_2$ as input, $\widetilde{\As}$ can be
  constructed in {\sc Exptime}.
\end{lem}

\begin{proof}
  Testing synchronizability of a set of states (resp.~a set of pairs
  of states) can easily be reduced to checking nonempty intersection
  between a set of NFAs. This is known to be a {\sc Pspace}-complete
  problem. It follows that computing the synchronizable sets can be
  done in {\sc Exptime}. It is then clear that the remaining
  computations can be done in {\sc Exptime}.
\end{proof}

We now prove that the construction is correct, \emph{i.e.}, that it
reduces \lt- and \ltt-separability to the special case $k=1$.

\begin{lem} \label{lem:reducor} Let $\As=(Q,A,\delta)$  be an NFA, and let
  $I_1,F_1,I_2,F_2\subseteq Q$. Then:
  \begin{enumerate}
  \item $L(\As,I_1,F_1),L(\As,I_2,F_2)$ are \lt-separable if and only if
    $L(\widetilde{\As},I_1,F_1),L(\widetilde{\As},I_2,F_2)$ are $\lt[1]$-separable.
  \item $L(\As,I_1,F_1),L(\As,I_2,F_2)$ are \ltt-separable if and only if there
    exists $d \in \nat$ such that $L(\widetilde{\As},I_1,F_1),L(\widetilde{\As},I_2,F_2)$
    are $\ltt[1,d]$-separable.
  \end{enumerate}
\end{lem}

\begin{proof}
  We only give the proof for \ltt using Condition~\eqref{item:f3} from
  Theorem~\ref{thm:sepltt}. The proof for \lt is obtained similarly
  using Condition~\eqref{item:e2} from Theorem~\ref{thm:seplt}.

  Suppose that $L(\As,I_1,F_1)$ and $L(\As,I_2,F_2)$ are not \ltt-separable. We prove that for
  all $d \in \nat$, $L(\widetilde{\As},I_1,F_1)$ and $L(\widetilde{\As},I_2,F_2)$ are not
  $\ltt[1,d]$-separable. Set $d \in \nat$. By Condition~\eqref{item:f3} 
  of Theorem~\ref{thm:sepltt}, there exists a common $d$-pattern \Ps:
  there are two words $w_1 \in L(\As,I_1,F_1)$ and $w_2 \in L(\As,I_2,F_2)$ such that $w_1 = u_0
  v_1 u_1 \ldots v_n u_n$ and $w_2 = u'_0 v'_1 u'_1 \ldots v'_m u'_m$ are
  \pfsdecomps, respectively compatible with some $(q_1,r_1)\in I_1\times F_1$ in
  $\As$ and $(q_2,r_2)\in I_2\times F_2$ in $\As$. There
  are two cases.

  In the first case, $n=m=0$ and $w_1=w_2$. Therefore, $L(\As,I_1,F_1) \cap
  L(\As,I_2,F_2) \ne \varnothing$. By definition, it follows that there exists
  some $b \in B_w$ such that $b \in L(\widetilde{\As},I_1,F_1) \cap
  L(\widetilde{\As},I_2,F_2)$, which ends the proof.

  In the second case, $n,m > 0$. Set $\widehat{w}_1 = \frp\frb_1 \cdots
  \frb_{n-1}\frs$ such that $\frp=(u_0,v_1)$, $\frs=(u_n,v_n)$ and for
  all $i$, $\frb_i =(v_i,u_i,v_{i+1})$. Similarly, we define $\widehat{w}_2 =
  \frp' \frb'_1 \cdots \frb'_{m-1}\frs'$. Observe that since we started
  from \pfsdecomps for a $d$-pattern \Ps, we have $\widehat{w}_1
  \ltteq{1}{d} \widehat{w}_2$. We use these words to construct $\widetilde{w}_1,
  \widetilde{w}_2 \in B^*$ belonging to $L(\widetilde{\As},I_1,F_1)$ and
  $L(\widetilde{\As},I_2,F_2)$, respectively.

  Let $\frb$ be a block appearing in $\widehat{w}_1,\widehat{w}_2$. We set
  $T_\frb \subseteq Q^2$ as the set of pairs of states
  that correspond to $\frb$ in the runs on $w_1$ and $w_2$. Since the
  \pfsdecomps are compatible with their respective NFAs, by definition,
  $T_\frb$ is synchronizable. Similarly, if $\frp$ (resp.~$\frs$) is a
  prefix (resp.~suffix) block occurring in $\widehat{w}_1,\widehat{w}_2$, we get
  a prefix (resp.~suffix) synchronizable set $T_\frp$ (resp.~$T_\frs$).
  We now set $\widetilde{w}_1=T_\frp T_{\frb_1} \cdots T_{\frb_{n-1}}
  T_\frs$ and $\widetilde{w}_2 = T_{\frp'} T_{\frb'_1} \cdots
  T_{\frb'_{m-1}}T_{\frs'}$. By definition of $\widetilde{\As}$
  and the fact that the \pfsdecomps are  $(q_1,r_1)$-, $(q_2,r_2)$-compatible,
  $\widetilde{w}_1,\widetilde{w}_2 \in B^*$ belong to  $L(\widetilde{\As},I_1,F_1)$ and
  $L(\widetilde{\As},I_2,F_2)$, respectively. Moreover, since $\widehat{w}_1 \ltteq{1}{d} \hat{w_2}$,
  we have $\widetilde{w}_1 \ltteq{1}{d} \widetilde{w}_2$. We conclude that
  $L(\widetilde{\As},I_1,F_1)$ and $L(\widetilde{\As},I_2,F_2)$ are not $\ltt[1,d]$-separable.

  \medskip
  Conversely, assume that for all $d \in \nat$,
  $L(\widetilde{\As},I_1,F_1)$ and $L(\widetilde{\As},I_2,F_2)$ are not $\ltt[1,d]$-separable. We prove that for all $d \in \nat$, there exists a pair in $(I_1\times
  F_1)\times(I_2\times F_2)$ having a common
  $d$-pattern \Ps. Set $d \in \nat$, by hypothesis, there exist
  $\widetilde{w}_1, \widetilde{w}_2 \in B^*$ in $L(\widetilde{\As},I_1,F_1),
  L(\widetilde{\As},I_2,F_2)$, respectively, such that $\widetilde{w}_1 \ltteq{1}{d}
  \widetilde{w}_2$. Again, we have two cases. By definition of
  $\widetilde{\As}$, either $\widetilde{w}_1,\widetilde{w}_2 \in B_w$
  or $\widetilde{w}_1,\widetilde{w}_2 \in B_p(B_i)^*B_s$. In the first case, this
  means that $L(\As,I_1,F_1)\cap L(\As,I_2,F_2)\neq\varnothing$, and therefore
  it suffices to set \Ps as a word in this intersection to conclude.

  Otherwise, $\widetilde{w}_1=b_pb_1\cdots b_nb_s$ and
  $\widetilde{w}_2=b_pb'_1\cdots b'_mb_s$. By definition of $B_p,B_i, B_s$,
  to each label appearing in $\widetilde{w}_1,\widetilde{w}_2$ we can associate
  a unique block, prefix block or suffix block. We set
  $\Ps=(\frp,f,\frs)$ as the $d$-pattern defined in the following
  way: \frp,\frs are the prefix and suffix blocks associated to $b_p$
  and $b_s$ respectively. Since $\widetilde{w}_1 \ltteq{1}{d} \widetilde{w}_2$,
  for all blocks $\frb$, the number of occurrences of labels $b$ such
  that \frb is associated to $b$ is the same in $\widetilde{w}_1,
  \widetilde{w}_2$ up to threshold $d$, we set $f(\frb)$ as this
  number. Finally, let $p_i\in I_i$ and $q_i\in F_i$ such that $\widetilde{w}_i$
  labels a path from $p_i$ to $q_i$ in $\widetilde{\As}$. It
  is now straightforward to verify that the pair $\bigl((p_1,q_1),(p_2,q_2)\bigr)$ admits compatible
  \pfsdecomps in $\As$, which terminates the proof.
\end{proof}

\subsubsection{Deciding the case $k=1$}

We explain how \lt- and \ltt-separability can be decided when the size
of profiles if fixed to $1$. Observe that in this case, the
\profile{1} of a position is its label. We prove the two following
lemmas.

\begin{lem} \label{lem:lt1} Deciding whether two regular languages given by
  an accepting NFA are $\lt[1]$-separable is in {\sc co-Np}.
\end{lem}

\begin{lem} \label{lem:ltt1} Given an NFA accepting two languages $L_1$ and
  $L_2$, deciding whether there exists $d \in \nat$ such that $L_1$ and $L_2$
  are $\ltt[1,d]$-separable is in {\sc Expspace}.
\end{lem}

As we explained in Section~\ref{sec:ltt}, decidability follows from
Parikh's Theorem and decidability of Presburger arithmetic. However,
applying these results naively yields high complexity. We explain
here how to refine the argument in order to get {\sc co-Np} and {\sc
Expspace} complexities. 

Set $\pi(L_1), \pi(L_2)$ as the Parikh's images of $L_1$ and
$L_2$, respectively. As we explained in Section~\ref{sec:ltt}, non-\ltt-separability is
equivalent to the following Presburger property: ``for all $d \in \nat$ there
exist $\bar{x}_1 \in \pi(L_1)$, $\bar{x}_2 \in \pi(L_2)$ that are equal
componentwise up to threshold~$d$.'' By~\cite{SeidlSMH04}, existential
Presburger formulas for $\pi(L_1),\pi(L_2)$ can be computed in linear time
(see also~\cite{VSS05} for the same technique applied to context-free
grammars). Therefore, a Presburger formula for the property can also be
computed in linear time. Moreover, by definition, this property has exactly
one quantifier alternation, it then follows from~\cite{ReddyLoveland:1978} that it
can be decided in {\sc Expspace}.

The same construction can be done for $\lt$. However, the counting
threshold $d$ is fixed to~$1$. Therefore, the constructed formula is
existential. It is known that existential Presburger formuals can be
decided in {\sc Np} (see~\cite{BoroshTreybig:76,gathen78}). We conclude that 
$\lt[1]$-separability is in~{\sc co-Np}. 

In the case of \lt the problem is actually {\sc co-Np}-complete. This
means that Lemma~\ref{lem:lt1} cannot be improved and that improving
Proposition~\ref{prop:ult} would require improving the reduction.  For
\ltt, the situation is different, it is likely that a sharper analysis
of the Presburger formula would yield a better upper bound. Indeed,
while deciding Presburger formulas with only one quantifier
alternation is already very costly in general, the formula we consider
is very specific.

\subsection{Lower Bounds}

\makeatletter{}%
\tikzstyle{nop}=[minimum size=0.35cm,draw,circle,inner sep=0pt]
\tikzstyle{nod}=[minimum size=0.35cm,draw,circle,inner sep=2pt]
\tikzstyle{nof}=[minimum size=0.35cm,draw,circle,double]
\tikzstyle{sqd}=[draw,rectangle,thick]

\tikzstyle{ars}=[line width=0.7pt,->]

\tikzstyle{arr}=[line width=2.0pt,double,->]
\tikzstyle{wrd}=[line width=0.25cm,gray!50,-]
\tikzstyle{sep}=[line width=1.0pt,-]
\tikzstyle{zig}=[line width=1.0pt,snake=zigzag]
\tikzstyle{ond}=[line width=1.0pt,snake=snake]
\tikzstyle{acc}=[line width=1.0pt,snake=brace]
\tikzstyle{stp}=[midway,draw,thick,rounded rectangle,align=center,fill=white]

\tikzstyle{bag}=[inner sep=0pt]
\tikzstyle{box}=[rectangle]

\tikzstyle{acc}=[line width=1.0pt,snake=brace]

In this subsection, we prove \textsc{co-Np} lower bounds for both \lt and
\ltt-separability. The bounds hold when the input languages are given
as NFAs or DFAs.

\begin{prop} \label{prop:hard} Let $L_1,L_2$ be languages accepted by
  two input DFAs $\mathcal{A}_1,\mathcal{A}_2$, respectively. The two following problems
  are \textsc{co-Np}-hard:
  
  \begin{enumerate}
  \item\label{itm:lt} Are $L_1$ and $L_2$ \lt-separable?
  \item\label{itm:ltt} Are $L_1$ and $L_2$ \ltt-separable?
  \end{enumerate}
\end{prop}

\begin{proof}
  We only do the proof for \ltt. The reduction is identical for the
  \lt case.  We prove that testing whether $L_1$ and $L_2$ are
  not \ltt-separable is {\sc Np}-hard. The proof is by reduction of {\bf
    3-SAT}. From an instance of {\bf 3-SAT}, we construct two DFAs and
  prove that the corresponding languages are not \ltt-separable if and
  only if the {\bf
    3-SAT} instance is satisfiable.

  Let $\Cs=\{C_1,\dots,C_m\}$ be a set of $3$-clauses over the set of
  variables $\{x_1,\ldots,x_n\}$. We construct DFAs $\As_1$ and $\As_2$
  over the alphabet $A := \{ \#, x_1, \ldots ,x_n, \neg x_1, \ldots,
  \neg x_n\}$. Given $w \in A^*$, we say that $w$ \emph{encodes an
    assignment of truth values} if for all $i \leq n$, $w$ contains either
  the label $x_i$ or the label $\neg x_i$, but not both.
  It is straightforward to see that an assignment of truth values for the
  variables can be uniquely defined from such a word. Moreover, by
  definition of \ltt, we have the following fact.

  \begin{fact} \label{fct:assign}
    The language of correct assignments is \ltt.
  \end{fact}

  \noindent
  Intuitively, we want our DFAs $\As_1,\As_2$ to verify the following
  three properties:

  \begin{enumerate}
  \item\label{item:code1} all words accepted by $\As_1$ encode assignments, and for each
    assignment of variables, a word coding that assignment is
    accepted by $\As_1$.
  \item\label{item:code2} all assignments accepted by $\As_2$ satisfy \Cs, and for each
    assignment of variables satisfying \Cs, a word coding that
    assignment is accepted by $\As_2$.
  \item\label{item:code3} if there exist $w_1,w_2\in A^*$ accepted by
    $\As_1,\As_2$ that encode the same assignment of variables, then for
    all $k,d \in \nat$, $w_1,w_2$ can be chosen such that $w_1 \kdltteq
    w_2$.
  \end{enumerate}

\noindent Conditions~\eqref{item:code1} and~\eqref{item:code2} are simple to enforce. Indeed, for Condition~\eqref{item:code1}, it
  suffices to construct a DFA recognizing the words $t_1 \dots t_n$
  such that for all $i \leq n$, $t_i=x_i$ or $t_i= \neg x_i$. For
  Condition~\eqref{item:code2} it suffices to construct a DFA recognizing the words $t_1
  \dots t_mu$ where for all $i \leq m$, $t_i$ is a literal of $C_i$ and
  $u$ is an arbitrary word ($u$ is necessary since some variables might
  not appear in the prefix, preventing the word from coding an
  assignment). The problem with these two template DFAs is that they
  do not verify Condition~\eqref{item:code3}. To solve this issue, we add loops on their
  states in order to generate as many copies of infixes as necessary to
  make two words coding the same assignment indistinguishable wrt.~the class~\ltt.

  We begin by giving $\As_1$, of which a graphical representation is
  shown in Figure~\ref{fig:a1}. It recognizes $L_1$, with the marked initial and final states.

  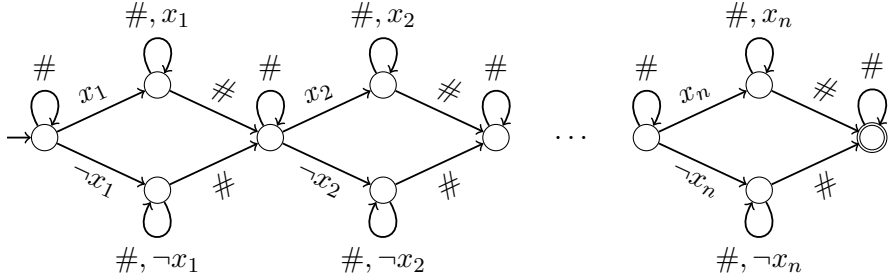
\begin{figure}[h]
    \begin{center}
      \begin{tikzpicture}
        \tikzset{every loop/.style={min distance=2.5mm,out=120,in=60,looseness=10}}

        \node[nod] (A10) at (0.0,+0.0) {};
        \node[nod] (A11) at (1.5,+0.7) {};
        \node[nod] (A12) at (1.5,-0.7) {};

        \node[nod] (A20) at (3.0,+0.0) {};
        \node[nod] (A21) at (4.5,+0.7) {};
        \node[nod] (A22) at (4.5,-0.7) {};

        \node[nod] (A30) at (6.0,+0.0) {};

        \node[nod] (A40) at (8.0,+0.0) {};
        \node[nod] (A41) at (9.5,+0.7) {};
        \node[nod] (A42) at (9.5,-0.7) {};

        \node[nof] (A50) at (11.0,+0.0) {};

        \draw[ars] ($(A10)-(0.5,0.0)$) to (A10);

        \draw[ars] (A10) to node[sloped,above] {$x_1$} (A11);
        \draw[ars] (A10) to node[sloped,below] {$\neg x_1$} (A12);

        \draw[ars] (A11) to node[sloped,above] {$\#$} (A20);
        \draw[ars] (A12) to node[sloped,below] {$\#$} (A20);

        \draw[ars] (A20) to node[sloped,above] {$x_2$} (A21);
        \draw[ars] (A20) to node[sloped,below] {$\neg x_2$} (A22);

        \draw[ars] (A21) to node[sloped,above] {$\#$} (A30);
        \draw[ars] (A22) to node[sloped,below] {$\#$} (A30);

        \draw[ars] (A40) to node[sloped,above] {$x_n$} (A41);
        \draw[ars] (A40) to node[sloped,below] {$\neg x_n$} (A42);

        \draw[ars] (A41) to node[sloped,above] {$\#$} (A50);
        \draw[ars] (A42) to node[sloped,below] {$\#$} (A50);

        \node at ($(A30)!1/2!(A40)$) {$\ldots$};

        \draw[ars] (A10) to [out=105,in=75,loop] node[above] {$\#$} (A10);
        \draw[ars] (A20) to [out=105,in=75,loop] node[above] {$\#$} (A20);
        \draw[ars] (A30) to [out=105,in=75,loop] node[above] {$\#$} (A30);
        \draw[ars] (A40) to [out=105,in=75,loop] node[above] {$\#$} (A40);
        \draw[ars] (A50) to [out=105,in=75,loop] node[above] {$\#$} (A50);

        \draw[ars] (A11) to [out=+105,in=+75,loop] node[above] {$\#,x_1$} (A11);
        \draw[ars] (A21) to [out=+105,in=+75,loop] node[above] {$\#,x_2$} (A21);

        \draw[ars] (A41) to [out=+105,in=+75,loop] node[above] {$\#,x_n$} (A41);

        \tikzset{every loop/.style={min distance=2.5mm,in=-120,out=-60,looseness=10}}        
        \draw[ars] (A12) to [out=-105,in=-75,loop] node[below] {$\#,\neg x_1$} (A12);
        \draw[ars] (A22) to [out=-105,in=-75,loop] node[below] {$\#,\neg x_2$} (A22);
        \draw[ars] (A42) to [out=-105,in=-75,loop] node[below] {$\#,\neg x_n$} (A42);
      \end{tikzpicture}
    \end{center}
    \caption{Representation of $\As_1$}
    \label{fig:a1}
  \end{figure}

  By definition, $\As_1$ verifies Condition~\eqref{item:code1}. We now define $\As_2$ as a sequence
  of $m$ subautomata, each one corresponding to a clause $C$ in
  $\Cs$. Intuitively, if $C=x_i \vee x_j \vee \neg x_k$, the subautomata
  selects a label within $\{x_i,x_j,\neg x_k\}$. This means that if this word
  encodes an assignment, it must satisfy all clauses in \Cs. We give a
  graphical representation of $\As_2$ in Figure~\ref{fig:a2} (it recognizes
  $L_2$, with the marked initial and final states).

  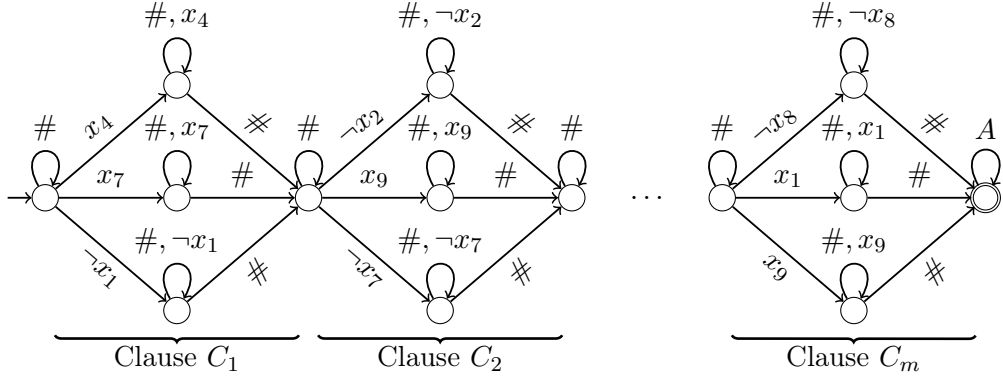
\begin{figure}[h]
    \begin{center}
      \begin{tikzpicture}
        \tikzset{every loop/.style={min distance=2.5mm,out=120,in=60,looseness=10}}
        
        \node[nod] (B1) at (0.0,0.0) {};
        \node[nod] (B2) at (3.5,0.0) {};
        \node[nod] (B3) at (7.0,0.0) {};
        \node[nod] (B4) at (9.0,0.0) {};
        \node[nof] (B5) at (12.5,0.0) {};

        \draw[ars] ($(B1)-(0.5,0.0)$) to (B1);

        \node at ($(B3)!1/2!(B4)$) {$\ldots$};

        \draw[ars] (B1) to [out=105,in=75,loop] node[above] {$\#$} (B1);
        \draw[ars] (B2) to [out=105,in=75,loop] node[above] {$\#$} (B2);
        \draw[ars] (B3) to [out=105,in=75,loop] node[above] {$\#$} (B3);
        \draw[ars] (B4) to [out=105,in=75,loop] node[above] {$\#$} (B4);
        \draw[ars] (B5) to [out=105,in=75,loop] node[above] {$A$} (B5);

        \node[nod] (C2) at ($(B1)!1/2!(B2)$) {};
        \node[nod] (C1) at ($(C2)+(0.0,1.5)$) {};
        \node[nod] (C3) at ($(C2)-(0.0,1.5)$) {};

        \draw[ars] (B1) to node[sloped,above] {$x_4$} (C1);
        \draw[ars] (B1) to node[sloped,above] {$x_7$} (C2);
        \draw[ars] (B1) to node[sloped,below] {$\neg x_1$} (C3);

        \draw[ars] (C1) to node[sloped,above] {$\#$} (B2);
        \draw[ars] (C2) to node[sloped,above] {$\#$} (B2);
        \draw[ars] (C3) to node[sloped,below] {$\#$} (B2);

        \draw[ars] (C1) to [out=+105,in=+75,loop] node[above] (R1) {$\#,x_4$} (C1);
        \draw[ars] (C2) to [out=+105,in=+75,loop] node[above] {$\#,x_7$} (C2);
        \draw[ars] (C3) to [out=+105,in=+75,loop] node[above] {$\#,\neg x_1$}
        (C3);

        \node[nod] (D2) at ($(B2)!1/2!(B3)$) {};
        \node[nod] (D1) at ($(D2)+(0.0,1.5)$) {};
        \node[nod] (D3) at ($(D2)-(0.0,1.5)$) {};

        \draw[ars] (B2) to node[sloped,above] {$\neg x_2$} (D1);
        \draw[ars] (B2) to node[sloped,above] {$x_9$} (D2);
        \draw[ars] (B2) to node[sloped,below] {$\neg x_7$} (D3);

        \draw[ars] (D1) to node[sloped,above] {$\#$} (B3);
        \draw[ars] (D2) to node[sloped,above] {$\#$} (B3);
        \draw[ars] (D3) to node[sloped,below] {$\#$} (B3);

        \draw[ars] (D1) to [out=+105,in=+75,loop] node[above] (R1) {$\#,\neg x_2$} (D1);
        \draw[ars] (D2) to [out=+105,in=+75,loop] node[above] {$\#,x_9$} (D2);
        \draw[ars] (D3) to [out=+105,in=+75,loop] node[above] {$\#,\neg x_7$} (D3);

        \node[nod] (F2) at ($(B4)!1/2!(B5)$) {};
        \node[nod] (F1) at ($(F2)+(0.0,1.5)$) {};
        \node[nod] (F3) at ($(F2)-(0.0,1.5)$) {};

        \draw[ars] (B4) to node[sloped,above] {$\neg x_8$} (F1);
        \draw[ars] (B4) to node[sloped,above] {$x_1$} (F2);
        \draw[ars] (B4) to node[sloped,below] {$x_9$} (F3);

        \draw[ars] (F1) to node[sloped,above] {$\#$} (B5);
        \draw[ars] (F2) to node[sloped,above] {$\#$} (B5);
        \draw[ars] (F3) to node[sloped,below] {$\#$} (B5);

        \draw[ars] (F1) to [out=+105,in=+75,loop] node[above] (R1) {$\#,\neg x_8$} (F1);
        \draw[ars] (F2) to [out=+105,in=+75,loop] node[above] {$\#,x_1$} (F2);
        \draw[ars] (F3) to [out=+105,in=+75,loop] node[above] {$\#,x_9$} (F3);

        \draw[acc] ($(B2.south west)-(0.0,1.7)$)  to node[below,sloped] {Clause $C_1$} ($(B1.south east)-(0.0,1.7)$);

        \draw[acc] ($(B3.south west)-(0.0,1.7)$)  to node[below,sloped] {Clause $C_2$} ($(B2.south east)-(0.0,1.7)$);

        \draw[acc] ($(B5.south west)-(0.0,1.7)$)  to node[below,sloped] {Clause $C_m$} ($(B4.south east)-(0.0,1.7)$);

      \end{tikzpicture}
    \end{center}
    \caption{Representation of $\As_2$, for $C_1=x_4 \vee x_7 \vee
      \neg x_1$, $C_2=\neg x_2 \vee x_9 \vee \neg x_7$, $C_m=\neg x_8 \vee x_1 \vee x_9$}
    \label{fig:a2}
  \end{figure}

  By definition, $\As_2$ verifies Condition~\eqref{item:code2}. It remains to verify that
  Condition~\eqref{item:code3} holds.

  \begin{lem} \label{lem:reduc-sat} If there are words $w_1,w_2$
    accepted by $\As_1,\As_2$ that encode the same assignment, then
    for all $k,d \in \nat$, they can be chosen such that $w_1 \kdltteq
    w_2$.
  \end{lem}

  \begin{proof}
    This is done by using the loops in $\As_1,\As_2$ to generate as many
    copies of the \kprofiles in $w_1,w_2$ as needed in order to get words
    that are $\kdltteq$-equivalent.
  \end{proof}

  We finish by proving that \Cs is satisfiable if and only if $L_1,L_2$
  are not \ltt-separable. Assume first that $L_1$ and
  $L_2$ are not \ltt-separable. By Fact~\ref{fct:assign},
  this means that there exist $w_1 \in L_1$ and $w_2 \in
  L_2$ sharing the same alphabet. By Condition~\eqref{item:code1}, $w_1$ encodes an
  assignment. Therefore, $w_2$ (which has the same alphabet) encodes the
  same assignment which satisfies \Cs by Condition~\eqref{item:code2}. Hence $\Cs$ is
  satisfiable.

  Conversely, assume that \Cs is satisfiable. We prove that for all $k,d
  \in \nat$,  $L_1,L_1$ are not $\ltt[k,d]$-separable. Set $k,d
  \in \nat$ and consider an assignment of truth values satisfying
  \Cs. By Conditions~\eqref{item:code1} and~\eqref{item:code2}, there must exist $w_1,w_2$ accepted by $\As_1,
  \As_2$ that both encode this assignment. It follows from
  Lemma~\ref{lem:reduc-sat} that $w_1,w_2$ can be chosen such that $w_1  
  \kdltteq w_2$. Therefore, $L_1,L_2$ are not
  $\ltt[k,d]$-separable, which terminates the proof.
\end{proof}

\section{The Case of Context-Free Languages}
\label{sec:cf}
\makeatletter{}%
In order to prove decidability of \ltt-separability for regular
languages, we needed three ingredients: Parikh's Theorem,
decidability of Presburger Arithmetic and Item~\eqref{item:b3} in
Theorem~\ref{thm:sepltt}. Since Parikh's Theorem holds not
only for regular languages but also for context-free languages, we
retain at least two of the ingredients in the context-free setting.

In particular, we can reuse the argument of Section~\ref{sec:ltt} to
prove that once the size $k$ of the infixes is fixed, separability by
\ltt is decidable for context-free languages. For any fixed $k \in
\nat$, we write $\ltt[k] = \bigcup_{d \in \nat} \ltt[k,d]$.

\begin{thm} \label{thm:lttkdecidcf}
  Let $L_1,L_2$ be context-free languages and $k \in \nat$. It is
  decidable whether $L_1,L_2$ are $\ltt[k]$-separable.
\end{thm}

An interesting consequence of Theorem~\ref{thm:lttkdecidcf} is that
$\ltt[1]$-separability of context-free languages is decidable. A
language is $\ltt[1]$ if and only if it can be defined by a
first-order logic formula that can only test equality between
positions, but not ordering. This result is surprising since
membership of a context-free language in this class is undecidable. We
give a proof of this fact below, which is a simple
adaptation of the proof of Greibach's Theorem (which is in particular
used to prove that regularity of a context-free language is~undecidable).

\begin{thm} \label{thm:lttkdecidcf2} Let $L$ be a context-free
  language. It is undecidable whether $L \in \ltt[1]$.
\end{thm}

\begin{proof}
  We reduce universality of context-free languages to this membership
  problem. Let $L$ be a context-free language over $A$ and let $\#
  \not\in A$. Let $K \not\in \ltt[1]$ be some context-free language
  and set $L_1 = (K \cdot \# \cdot A^{*}) \cup (A^{*} \cdot \# \cdot
  L)$. Clearly, a context-free grammar for $L_1$ can be computed from
  a context-free grammar for $L$. We show that $L=A^{*}$ if and only
  if~$L_1 \in \ltt[1]$.

  If $L=A^{*}$, then $L_1= A^{*} \cdot \# \cdot A^{*} \in
  \ltt[1]$. Conversely, assume that $L_1 \in \ltt[1]$, and suppose by
  contradiction that $L \neq A^{*}$. Pick $w \in A^{*}$ such that $w
  \not\in L$. By definition, $K=\{u \mid u \# w \in L_1\}$. One can
  verify that $\ltt[1]$ is closed under right residuals. Therefore,
  $K=L_1(\#w)^{-1} \in \ltt[1]$ which is a contradiction by definition
  of $K$.
\end{proof}

These two results may seem contradictory. Indeed in the setting of
regular languages, membership can be reduced to separability (a
language belongs to a class if the class can separate it from its
complement). However, context-free languages are not closed under
complement, which makes the reduction false in this larger setting.

An interesting question is whether decidability extends to full \lt-
and \ltt-separability of context-free languages. This would also be
surprising since membership of a context-free language in \lt or \ltt
are undecidable problems. Such a result would require to
generalize our third ingredient,  Item~\eqref{item:b3} in
Theorem~\ref{thm:sepltt}, to context-free languages. This means that
we would need a method for computing a bound on the size of the
infixes that a potential separator has to consider. It turns out that
this is not possible. 

\begin{thm} \label{thm:lttkdecidcf3}
  Let $L_1,L_2$ be context-free languages. It is undecidable 
  whether $L_1,L_2$ are \lt-separable. It is undecidable 
  whether
  $L_1,L_2$ are \ltt-separable.
\end{thm}

It was already known~\cite{szygram} that separability by a regular
language is undecidable for context-free languages. The proof of
Theorem~\ref{thm:lttkdecidcf3} is essentially the same since the
reduction provided in~\cite{szygram} actually works for any class of
regular separators that contains all languages of the form $K_1A^{*}\cup
K_2$ where $K_1,K_2$ are finite languages. Since this is
clearly the case for both \lt and \ltt, Theorem~\ref{thm:lttkdecidcf3}
follows. For the sake of completeness, we provide a version of this proof
tailored to \lt and \ltt below.

\begin{proof}[Proof of Theorem~\ref{thm:lttkdecidcf3}]
  The proof is done by reduction of the halting problem on Turing
  machines to \lt-separability and \ltt-separability. The reduction is
  the same for both \lt and \ltt and is essentially a rewriting of a
  proof of~\cite{szygram}.

  Consider a deterministic Turing machine $\mathcal{M}$. We prove that it is
  possible to compute context-free languages $L_1,L_2$ from $\mathcal{M}$ such
  that $\mathcal{M}$ halts on the empty input if and only if $L_1,L_2$ are
  \lt-separable, if and only if $L_1,L_2$ are \ltt-separable.

  Let $A$ be the alphabet of $\mathcal{M}$, let $Q$ be its set of
  states, and let $B=A \cup (A \times Q) \cup \{\#,\gamma\}$ where
  $\#,\gamma \not\in A$. As usual, we encode configurations of
  $\mathcal{M}$ as words in $A^{*} \cdot (A \times Q) \cdot A^{*}
  \subseteq B^*$: the word $(u,(q,a),v)$ means that $\mathcal{M}$ is
  in state $q$, the tape holds $u\cdot a\cdot v$, and the head
  currently scans the distinguished $a$ position. Finally, if $w \in
  B^*$, we denote by $w^{R}$ the mirror image of $w$. We can now
  define the context-free languages $L_1,L_2$ over $B$. The language $L_1$
  contains all words of the form:
  \[
  c^{}_1\#c_2^{R}\#c^{}_3\#c_4^R \cdots \#c^{}_{2k-1}\#c_{2k}^R \# \gamma^k
  \]
  such that $c_1,\dots,c_{2k}$ are encodings of configurations of $\mathcal{M}$,
  and for all $i \leq k$, $c_{2i-1}\vdash_{\mathcal{M}} c_{2i}$ (\emph{i.e.},
  $c_{2i}$ is the configuration of $\mathcal{M}$ reached after one computation
  step from configuration $c_{2i-1}$). Similarly, $L_2$ contains all
  words of the form:
  \[
  c^{}_1\#c_2^{R}\#c^{}_3\#c_4^R \cdots \#c^{}_{2k-1}\#c_{2k}^R \# \gamma^{2k}
  \]
  such that $c_1,\dots,c_{2k}$ are encodings of configurations of
  $\mathcal{M}$, $c_1$ is the initial configuration of $\mathcal{M}$
  starting with an empty input and for all $i \leq k-1$,
  $c_{2i}\vdash_{\mathcal{M}}c_{2i+1}$. It is classical that $L_1,L_2$
  are indeed context-free and that grammars for $L_1,L_2$ can be
  computed from $\mathcal{M}$. We now make a simple observation about
  prefixes that are common to both languages. Let
  $c_1,c_2,c_3,c_4,\ldots,c_{i-1},c_{i}$ be a sequence of
  configurations and let $w\in B^{*}$ be the word
  \[
  \begin{array}{lcll}
    w & = & c^{}_1\#c_2^{R}\#c^{}_3\#c_4^R \cdots \#c^{}_{2k-1}\#c_{2k}^R\#
    & \text{(if $i=2k$),} \\[.8ex]
    w & = & c^{}_1\#c_2^{R}\#c^{}_3\#c_4^R \cdots
    \#c^{}_{2k-1}\#c_{2k}^{R}\#c^{}_{2k+1}\# & \text{(if $i=2k+1$).}
  \end{array}
  \]
  By definition of $L_1,L_2$, we have the following fact:

   \begin{fact} \label{fct:reduc} If $w$ is both a prefix of some word
    in $L_1$ and some word in $L_2$, then $c_1,c_2,\ldots,c_{i}$ are
    the first $i$ configurations of the run of $\mathcal{M}$ starting from the
    empty input. Moreover, if $i=2k$ and
    \[
    c^{}_1\#c_2^{R}\#c^{}_3\#c_4^R \cdots \#c^{}_{2k-1}\#c_{2k}^R\#c\#
    \]
  is a prefix of a word in $L_2$, then $c$ is configuration $(i+1)$ in
  the run. Symmetrically, if $i=2k+1$ and
  \[
  c^{}_1\#c_2^{R}\#c^{}_3\#c_4^R \cdots \#c^{}_{2k-1}\#c_{2k}^{R}\#c^{}_{2k+1}\#c^{R}\#
  \]
  is a prefix of a word in $L_1$, then $c$ is configuration $(i+1)$ in
  the run.
\end{fact}

It remains to prove that this is indeed a reduction, \emph{i.e.}, that
$\mathcal{M}$ halts on the empty input if and only if $L_1,L_2$ are \lt-separable, if
and only if $L_1,L_2$ are \ltt-separable.
Assume first that $\mathcal{M}$ does not halt on empty input. This means that
the run of $\mathcal{M}$ is an infinite sequence of configurations
$c_1,c_2,c_3,\dots$. By definition of $L_1,L_2$, for all $k \in \nat$:
\[
\begin{array}{ll}
  c^{}_1\#c_2^{R}\#c^{}_3\#c_4^R \cdots \#c^{}_{2k-1}\#c_{2k}^R \# \gamma^k &\in L_1, \\[.8ex]
  c^{}_1\#c_2^{R}\#c^{}_3\#c_4^R \cdots \#c^{}_{2k-1}\#c_{2k}^R \# \gamma^{2k} &\in L_2.
\end{array}
\]
It is then easy to deduce that $L_1,L_2$ cannot be separated by a \lt
or \ltt language (actually not even by a regular language).

\smallskip

Conversely, assume that $\mathcal{M}$ halts on the empty input within $\ell$
 steps, \emph{i.e.},
$c_1\vdash_{\mathcal{M}}c_2\vdash_{\mathcal{M}}\ldots\vdash_{\mathcal{M}}c_\ell$
and $c_\ell$ is the halting configuration. Before defining an \lt
separator, we observe that sufficiently long prefixes of sufficiently long words of $L_1,L_2$
are~distinct.

\begin{lem} \label{lem:reduc} Let $w_1 \in L_1$ and $w_2\in L_2$ of length
  greater than $(\ell+1)^2+2(\ell+1)$ and let $u_1,u_2$ be the prefixes of length
  $\ell(\ell+1)+2\ell+1$ of $w_1,w_2$, respectively. Then $u_1 \neq u_2$.
\end{lem}

\begin{proof}
  We proceed by contradiction. Assume $u_1 = u_2$ and let $u$ be the
  largest prefix of $u_1=u_2$ of the form
  \[
  \begin{array}{l@{\;}ll}
    u=&c^{}_1 \# c_2^{R} \# c^{}_3 \# c_4^R \cdots \# c^{}_{i-1} \#
    c_{i}^R \#,&\text{or} \\[.8ex]
    u=&c^{}_1 \# c_2^{R} \# c^{}_3 \# c_4^R \cdots \# c_{i-1}^{R} \#
    c_{i} \#&\text{depending on whether $i$ is even or not.}
  \end{array}
  \]
  By definition $u_1=u_2=u \cdot v$, where $v$ can be of the form
  $\gamma^{j}$ with $j \leq i$ or is the prefix of $c$ or $c^R$ for
  some configuration $c$ of $\mathcal{M}$.

  Assume first that $v=\gamma^j$ with $j\leq i$. By
  Fact~\ref{fct:reduc}, $c_1,c_2,c_3,c_4,\ldots,c_{i-1},c_{i}$ are the
  first $i$ configurations of the run of $\mathcal{M}$ starting from the empty
  input. Since $\mathcal{M}$ halts in $\ell$ steps, this means that $i \leq \ell$ and
  that each configuration $c_i$ is of length at most $\ell+1$. It
  follows that $u$ is of length at most $\ell(\ell+1) + \ell$. Therefore,
  $u_1=u_2$ is of length at most $\ell(\ell+1)+ \ell+j \leq\ell(\ell+1)+ \ell+i \leq
  \ell(\ell+1) + 2\ell$, in contradiction with the definition of $u_1,u_2$.

  Assume finally that $v$ is the prefix of $c$ or $c^R$ for some configuration $c$ of
  $\mathcal{M}$. By Fact~\ref{fct:reduc}, $c$ is so that
  $c_1,c_2,c_3,c_4,\ldots,c_{i-1},c_{i},c$ are the first $i+1$
  configurations of the run of $\mathcal{M}$ starting from the empty
  input. Since $\mathcal{M}$ halts in $\ell$ steps, this means that $i+1 \leq \ell$
  and that each configuration is of length at most $\leq \ell+1$. It
  follows that $u_1=u_2$ is of length at most $\ell(\ell+1)+ \ell$ which is
  again a contradiction.
\end{proof}

Let $K_1$ be the language of words of $L_1$ of length less than
$\ell(\ell+1)+2\ell+1$. Similarly, let $K_2$ be the set of prefixes of length
$\ell(\ell+1)+2\ell+1$ of words in $L_1$. Finally, set $L= K_1 \cup K_2 \cdot
B^{*}$. By definition, $K_1,K_2$ are finite languages, hence $L$ is
clearly \lt (and therefore \ltt). It remains to prove that $L$ is a
separator.

By definition, we have $L_1 \subseteq L$. Let us prove that $L \cap
L_2 = \varnothing$. We proceed by contradiction: assume that there
exists $w \in L \cap L_2$. If $w \in K_1$, then the contradiction is
immediate because, by definition of $L_1$, $w$ must be in of the form
\[
c^{}_1\#c_2^{R}\#c^{}_3\#c_4^R \cdots \#c^{}_{2k-1}\#c_{2k}^R \# \gamma^k
\]
Now, belonging to $L_2$ would require twice as many letters $\gamma$ at the end
of the word. On the other hand, if $w \in K_2 \cdot B^{*}$, then, by definition of $K_2$
there exists some word of length $\ell(\ell+1)+2\ell+1$ that is both a prefix
of a word in $L_1$ and a prefix of a word in $L_2$. This contradicts
Lemma~\ref{lem:reduc}.

We deduce that $L_1,L_2$ are \lt-separable, and therefore also
\ltt-separable, which concludes the proof of
Theorem~\ref{thm:lttkdecidcf3}.
\end{proof}

\section{Conclusion}
\label{sec:conc}
\makeatletter{}%
We have shown separation theorems for both \lt and \ltt. In both cases,
we provide a decision procedure to test separability,
running in {\sc co-Nexptime} and {\sc 2-Expspace}
respectively. Another contribution is a description of possible
separators, given by bounds defining them.

Several questions remain open in this line of research. A first one is
to obtain tight complexity bounds for both classes. While we have {\sc
  co-Nexptime} and {\sc 2-Expspace} upper bounds for \lt and \ltt
respectively, we have only {\sc co-NP} lower bounds.  The upper bounds
rely on a reduction to the case $k=1$, \emph{i.e.}, a translation to
the special case when the size of infixes is fixed to $1$. This translation is exponential wrt.\
the size of the input automata. Improving the upper bounds would likely require improving
this reduction.

Another question is to generalize our techniques to other settings, and to
obtain transfer results. First, one can consider other fragments for
separability.  For instance, if separation is decidable for some fragment of
first-order logic, is it still decidable when adding some predicate to the
logic, such as the successor? In the present paper, this is what we do for the
specific case of \ltt, whose corresponding logic, \fos, is obtained by adding the successor to
first-order logic with only equality on positions and alphabetic
predicates. Another natural example is adding modulo predicates, which would
allow to treat \lttm for instance, the generalization of \ltt in which infixes
can now also be counted modulo constants. Generalizing our results to more
complex structures such as trees would also be interesting. However, in the
setting of trees, while decidable characterizations are known for both \lt and
\ltt~\cite{bsltt,pslt}, no delay theorem is known. This makes separation a
challenging problem as our techniques rely on a generalization of this
theorem.

\mypar{Acknowledgement.} We thank the referees for their careful reading and
helpful suggestions.

\bibliographystyle{abbrv}%
%

 %

%\newpage
%\appendix

\end{document}